\documentclass[11pt,a4paper]{article}
\usepackage{caption} 
\usepackage{algorithm}
\usepackage{algorithmic}

\usepackage{authblk}
\usepackage{fullpage}

\usepackage{newfloat}
\usepackage{listings}
\DeclareCaptionStyle{ruled}{labelfont=normalfont,labelsep=colon,strut=off} 
\floatstyle{ruled}
\newfloat{listing}{tb}{lst}{}
\floatname{listing}{Listing}

\usepackage{adjustbox}

\usepackage{amsmath,amsthm,amssymb}
\usepackage{mathtools}
\usepackage{xspace}
\usepackage{graphicx}
\usepackage{subcaption}
\usepackage{cleveref}
\RequirePackage{tikz}
	\usetikzlibrary{calc,fit,arrows.meta}
\usepackage{tabularx, booktabs}
\usepackage{makecell}
\usepackage{pifont}%
\newcommand{\cmark}{\ding{51}}%
\newcommand{\xmark}{\ding{55}}%

\usepackage{todonotes}
\usepackage{complexity} 

\usepackage{thmtools} 
\usepackage{thm-restate}

\usepackage{etoolbox}

\newtheorem{theorem}{Theorem}[section]
\newtheorem{proposition}[theorem]{Proposition}

\newtheorem{lemma}[theorem]{Lemma}
\newtheorem{corollary}[theorem]{Corollary}
\newtheorem{observation}[theorem]{Observation}
\newtheorem{claim}[theorem]{Claim}

\theoremstyle{definition}

\newtheorem{example}[theorem]{Example}

\RequirePackage[round]{natbib}
\newcommand{\ag}{N} 
\newcommand{\Partitions}{\Pi_N}
\newcommand{\caf}{CAF} 
\newcommand{\cafs}{CAFs}

\usepackage{pgfplots}
\usepackage{siunitx}

\pgfplotsset{compat=1.3}
\title{Causes of Stability in Dynamic Coalition Formation}

\author[1]{Niclas Boehmer}
\author[2]{Martin Bullinger}
\author[3]{Anna Maria Kerkmann}

\affil[1]{ \small Algorithmics and Computational Complexity, Technische Universit\"at Berlin}
\affil[2]{ \small School of Computation, Information and Technology, Technische Universit\"at M\"unchen}
\affil[3]{ \small Institut f\"ur Informatik, Heinrich-Heine-Universit\"at D\"usseldorf\protect\\ \vspace*{0.05cm} niclas.boehmer@tu-berlin.de, martin.bullinger@in.tum.de, anna.kerkmann@hhu.de}

\begin{document}

\maketitle

\begin{abstract}
	We study the formation of stable outcomes via simple dynamics in cardinal hedonic games, where the utilities of agents change over time depending on the history of the coalition formation process.
	Specifically, we analyze situations where members of a coalition decrease their utility for a leaving agent (\emph{resent}) or increase their utility for a joining agent (\emph{appreciation}).
	We show that in contrast to classical dynamics, for resentful or appreciative agents, dynamics are guaranteed to converge under mild conditions for various stability concepts.
	Thereby, we establish that both resent and appreciation are strong stability-driving forces. 
\end{abstract}

\section{Introduction}

Coalition formation is a vibrant topic in multi-agent systems that has been continuously researched during the last  decades.
It concerns the question of dividing a set of agents, for example, humans or machines, into disjoint coalitions such as research teams. 
Agents carry preferences over these coalition structures. 
A common assumption is that externalities, that is, the coalition structure outside one's own coalition, play no role.
This is captured in the prominent framework of hedonic games.
Moreover, the desirability of a coalition structure is usually measured with respect to stability.
Abstractly speaking, a coalition structure is stable if there is no agent or set of agents that can perform a beneficial deviation by joining existing coalitions or by forming new coalitions. 

There are two specific properties of hedonic games crucially influencing past research. 
First, the number of possible coalitions an agent can be part of is exponentially large.
Therefore, a repeatedly considered challenge is to come up with reasonable succinctly representable settings. 
It is very prominent in this context to aggregate utilities from cardinal valuations of other agents.
Second, most established stability concepts suffer from non-existence under strong restrictions which often leads to computational boundaries such as hardness of the decision problem whether a stable state exists. 
Much of the research has therefore focused on identifying suitable conditions guaranteeing stable states.

The dominant coalition formation framework is static in two dimensions. First, stability is usually a static concept in the sense that, since a coalition structure is either stable or not, we are only interested in \emph{finding} these stable structures. 
The underlying assumption here is that we operate in a centralized system where a (desirable) coalition structure can be created by a central authority. 
This paradigm has only recently been complemented by interpreting deviations of agents as a dynamic process. 
The goal here is to reach stable coalition structures through decentralized individual decisions \citep{BBW21a,BBT22a}.
Second, utility functions are static. 
To demonstrate the implications of this assumption, we describe a run-and-chase example, which is present in many classes of hedonic games. 
Consider a situation where there are only the two agents Alice and Bob. 
Alice wants to be alone in her coalition, whereas Bob wants to be in a joint coalition with Alice. 
It is clear that in the two possible coalition structures, there is always an agent who wants to change their situation. 
From a centralized perspective, this simply means that no coalition structure has the prospect of stability. 
In a distributed, dynamic setting where utilities are static, the following occurs indefinitely: Whenever Alice and Bob are in a joint coalition, then Alice leaves the coalition to be alone. 
However, whenever Alice and Bob are in two separate coalitions, then Bob joins Alice. 
In practice, such an infinite situation is unreasonable: 
After playing run-and-chase for a while, either Alice or Bob are likely to change their behavior and therefore their preferences.
On the one hand, Bob might get frustrated  because he is constantly left by Alice and therefore stops his efforts to join her. 
On the other hand, Alice could realize the high effort that Bob makes to be in a coalition with her and feels sufficient appreciation to eventually accept Bob in her coalition. 
In both scenarios, we reach a state that is stable because of the \emph{history} of the coalition formation process.

In this paper, we model situations where the history influences the agents' utilities, offering a new perspective on the reachability of stable coalition structures.
We study a dynamic coalition formation process where agents perform deviations based on stability concepts.
However, in contrast to previous work on dynamics, we assume that a deviation has an effect on the \emph{perception} of the deviator, resulting in agents changing their utility for the deviator.
We distinguish two approaches.
First, an agent might act \emph{resentfully} in the sense that, like Bob, she lowers her utility for an agent abandoning her. 
A deviator abandoning a resentful agent again and again eventually looses all of her attraction to the resentful agent.
On the other hand, an agent could \emph{appreciate} the effort of another agent to be part of her coalition, and therefore, like Alice, increase her utility for an agent whenever the agent joins her.
After sufficient effort, the urge to leave the deviator ceases.

\begin{table}[t!]\centering
	\caption{Overview of results. ``\cmark'' means that the corresponding dynamics is guaranteed to converge; ``\xmark'' means that we have an example for an infinite sequence. See \Cref{se:prelims} for definitions. For each result, we include the number of the respective statement.}\label{se:ov-table}
	\setlength{\tabcolsep}{4pt}
		\begin{tabular}{lc|c|c|c|c|c|c}
			\toprule
			& \multicolumn{1}{c}{SCS} & \multicolumn{1}{c}{CS} & \multicolumn{1}{c}{IS} & \multicolumn{1}{c}{CNS} & \multicolumn{1}{c}{NS} \\
			\midrule
			\makebox[-1pt][l]{\textbf{ASHG}}\\
			resent & \cmark~(\ref{thm:resent:SC-converges}) & \cmark~(\ref{thm:resent:SC-converges})& \cmark~(\ref{thm:resent:SC-converges})& \textbf{?} & \textbf{?} \\
			resent+IR & \cmark~(\ref{thm:resent:SC-converges})& \cmark~(\ref{thm:resent:SC-converges})& \cmark~(\ref{thm:resent:SC-converges})& \cmark~(\ref{thm:resent:NS-IR-converges})& \cmark~(\ref{thm:resent:NS-IR-converges}) \\
			appreciation & \xmark~(\ref{th:appreciative-core-cylce}) & \xmark~(\ref{th:appreciative-core-cylce}) & \textbf{?} & \cmark~(\ref{thm:apprec-CNS})& \textbf{?} \\
			\midrule
			\makebox[-1pt][l]{\textbf{MFHG}}\\
			resent &   \textbf{?} &  \textbf{?} &  \textbf{?} & \textbf{?} & \xmark~(\ref{th:mfhg-resent})  \\
			resent+IR & \cmark~(\ref{thm:resent:SC-can-converge}) & \cmark~(\ref{thm:resent:SC-can-converge})& \cmark~(\ref{thm:resent:SC-can-converge})& \cmark~(\ref{thm:resent:NS-IR-converges})& \cmark~(\ref{thm:resent:NS-IR-converges}) \\
			appreciation & \xmark~(\ref{th:appreciative-core-cylce}) & \xmark~(\ref{th:appreciative-core-cylce})& \textbf{?} & \textbf{?} & \xmark~(\ref{th:mfhg-apprec})\\
			\bottomrule
		\end{tabular}
\end{table}

\subsection{Contribution}

We initiate the study of cardinal hedonic games under utility functions changing over time. 
In particular, we consider utility modifications based on the resentful and appreciative perception of other agents. 
We investigate whether decentralized dynamics based on various types of deviations are guaranteed to converge.
Deviations might be constrained to be individually rational (IR), that is, a deviating agent needs to prefer her new coalition to being alone.
We showcase our results by considering additively separable hedonic games (ASHGs) and modified fractional hedonic games (MFHGs), where an agent's utility for a coalition is the sum or average utility for the other agents in the coalition, respectively.
\Cref{se:ov-table} provides an overview of these results.
First, for resentful agents performing individually rational deviations, convergence is guaranteed in all considered cases. 
If deviations may also violate individual rationality,
the situation becomes more complicated and elusive to a complete understanding; nevertheless, we establish several convergence guarantees while also having an involved example of a cycling dynamics in MFHGs.
In contrast, appreciation is usually not sufficient to guarantee convergence.
Notably, as proved in \Cref{co:equivalence} four of our open questions concerning both resentful and appreciative agents are in some sense equivalent. 
 
In fact, most of our results do not only apply to ASHGs and/or MFHGs but to larger classes of hedonic games. 
For this, we develop an axiomatic framework for utility aggregation  based on the perception of friends and enemies, that is, agents yielding positive and negative utility, respectively.

In our simulations, we observe that our model of dynamic utilities leads to the (quick) convergence of Nash dynamics. 
Moreover, we analyze the structure and expressiveness 
of the produced outcomes.
Finally, we outline results for other perception models and for computational questions concerned with finding shortest converging sequences.

\subsection{Related Work}

Hedonic games originate from economic theory \citep{DrGr80a}, but their constant and broad consideration only started with key publications by \citet{BKS01a}, \citet{CeRo01a}, and \citet{BoJa02a}.
An overview of hedonic games is provided in the survey by \citet{AzSa15a}.
The search for suitable representations of reasonable classes of hedonic games has led to various proposals \citep[see, e.g.,][]{CeRo01a,BoJa02a,Ball04a,ElWo09a,Olse12a,ABB+17a}.

Various stability concepts and their computational boundaries have been previously studied. We focus on results concerning ASHGs \citep{BoJa02a} and MFHGs \citep{Olse12a}.
\citet{SuDi10a} show prototype \NP-hardness reductions for single-agent stability concepts in ASHGs, paving the way for many similar results for single-agent and group stability \citep[see, e.g.,][]{ABS11c,BBT22a,Bull22a}. 
\citet{GaSa19a} consider ASHGs under symmetric utilities and show \PLS-completeness of computing stable states, while \citet{Woeg13a} and \citet{pet:c:precise-hed-games} show $\Sigma_2^\P$-completeness of the (strict) core in ASHGs. 
\citet{pet-elk:c:simple-causes-complex-hed-games} provide a meta view on computational hardness.
For MFHGs, there seem to be less computational boundaries. Indeed, for symmetric and binary utilities, stable states exist and can be efficiently computed. Core stability is even tractable for symmetric and arbitrarily weighted utilities \citep{MMV18a}.
Apart from the consideration of stability, other desirable notions of efficiency or fairness such as Pareto optimality, envy-freeness, or popularity have been studied for ASHGs and MFHGs \citep{ABS11c,EFF20a,Bull19a,BrBu22a}. 
These papers provide more evidence that MFHGs seem to be less complex than ASHGs.

The dynamical, distributed approach to coalition formation received increased attention very recently \citep{HVW18a,BFFMM18a,CMM19a,BBW21a,FMM21a,BBT22a,BMM22a}. 
There, \citet{BFFMM18a,BBW21a,BBT22a} consider stability based on single-agent deviations, whereas \citet{CMM19a,FMM21a} consider group stability.

 Finally, another recent approach to achieve a guarantee to stability in hedonic games is the consideration of loyalty \citep{BuKo21a} which is based on altruism in hedonic games \citep{KNRRRSW22}.

\section{Preliminaries and Model}\label{se:prelims}

In this section, we define the basic coalition formation setting, our specific model, and provide some first observations. For an integer $i\in \mathbb{N}$, we define $[i]=\{1,\ldots,i\}$.

\subsection{Cardinal Hedonic Games}
Let $N=[n]$ 
be a finite set of \emph{agents}.
A \emph{coalition} is any subset of $N$.
We denote the set of all possible coalitions containing agent $i\in N$ by
$\mathcal N_i = \{C\subseteq N\colon i\in C\}$.
Any \emph{partition} of the agents $N$ 
is also called \emph{coalition structure} and
we denote the set of all partitions of $N$ by $\Partitions$.
Given an agent $i\in N$ and a partition $\pi \in \Partitions$, let $\pi(i)$ denote the coalition of $i$, i.e., the unique coalition $C\in \pi$ with $i\in C$.
A \emph{(cardinal) hedonic game} is a pair $(N,u)$ consisting of
a set $N$ of agents and a utility profile $u=(u_i)_{i\in N}$ where $u_i\colon N \to \mathbb Q$ is the \emph{utility function} of agent $i$. 
Thus, for $i,j\in N$, $u_i(j)$ is $i$'s utility for agent $j$.
We sometimes equivalently view a utility function as a vector $u_i\in \mathbb Q^n$. 
An agent $j\in N$ is a \emph{friend} (or \emph{enemy}) of an agent $i\in N$ if $u_i(j) > 0$ (or $u_i(j) < 0$).

To move from utilities for single agents to utilities over coalitions, we use
\emph{cardinal aggregation functions} (\cafs).  
For every agent $i\in N$,
the {\caf} $A_i\colon \mathcal N_i \times \mathbb Q^n \to \mathbb Q$ specifies $i$'s utility for a given coalition 
for her given utility vector. 
Then, the utility of an agent for a partition $\pi$ with respect to aggregation function $A_i$ is $u_i^{A_i}(\pi) = A_i(\pi(i), u_i)$. To keep notation concise, we sometimes omit the {\caf} as a superscript when it is clear from the context.
For an agent $i\in N$ with utility function $u_i$, a coalition $C\in \mathcal{N}_i$ is \emph{individually rational} (IR) if $A_i(C,u_i) \ge A_i(\{i\},u_i)$. 
Further, a partition $\pi$ is  \emph{individually rational} (IR) for agent $i$ if $\pi(i)$ is an individually rational coalition.

Common classes of cardinal hedonic games such as the two specific classes studied in this paper have a straightforward representation with respect to \cafs. For each agent $i\in N$ with utility function $u_i$,
\begin{itemize}
	\item additively separable hedonic games (ASHGs) \citep{BoJa02a}
		use the aggregation function $\mathit{AS}$ defined by
		$\mathit{AS}_i(C, u_i)= \sum_{j\in C\setminus\{i\}}u_i(j)$ and
	\item modified fractional hedonic games (MFHGs) \citep{Olse12a}
		use the aggregation function $\mathit{MF}$ defined by
		$\mathit{MF}_i(C, u_i)= \frac{\sum_{j\in C\setminus \{i\}}u_i(j)}{|C|-1}$ if $|C|\ge 2$ and $\mathit{MF}_i(C, u_i) = 0$, otherwise.
\end{itemize}

\subsection{Deviations and Stability}

As indicated in the introduction, we 
distinguish different stability notions
based on single-agent deviations
and group deviations.
Given a partition $\pi\in \Partitions$,
agent $i\in N$ might perform 
a \emph{single-agent deviation} from $\pi(i)$ to any coalition $C\in \pi\cup\{\emptyset\}$,
resulting in the partition 
$\pi'
=(\pi \setminus \{\pi(i), C \})\cup \{\pi(i)\setminus\{i\},C\cup\{i\}\}$;
and a group of agents $C\subseteq N$ might perform
a \emph{group deviation},
leading to the partition 
$\pi'
=(\pi \setminus\{\pi(j)\mid j\in C \})\cup\{\pi(j)\setminus C\mid j\in C\}\cup\{C\}$.

Depending on which agents improve as a result of a deviation,
we distinguish the following \emph{types of deviations}. 
Agent $i$'s single-agent deviation from $\pi(i)$ to $C\in \pi\cup\{\emptyset\}$, 
resulting in partition $\pi'$,
is a
\emph{Nash} (NS) deviation if $u_i(\pi') > u_i(\pi)$.
An NS deviation of $i$ from $\pi$ to $\pi'$
is called
\begin{itemize}
	\item an \emph{individual} (IS) deviation if $u_j(\pi')\ge u_j(\pi)$
	for all $j\in C$, where $C$ is the coalition to which $i$ deviated;~and
	\item a \emph{contractual Nash} (CNS) deviation if $u_j(\pi')\ge u_j(\pi)$
	for all $j\in \pi(i)\setminus\{i\}$.
\end{itemize}
A group deviation of coalition $C$ from $\pi$ to $\pi'$ is 
\begin{itemize}
	\item a \emph{core} (CS) deviation if $u_i(\pi') > u_i(\pi)$
	for all $i\in C$; and
	\item a \emph{strict core} (SCS) deviation if $u_i(\pi') \ge u_i(\pi)$
	for all $i\in C$ and
	$u_j(\pi') > u_j(\pi)$
	for some $j\in C$.
\end{itemize}

Finally, for all types of deviations introduced above, we 
define the  respective stability notion of a partition by the absence of a corresponding deviation.
For example, a partition $\pi$ is said to be Nash-stable (NS) if 
there is no NS deviation from $\pi$ to another partition.
The logical relations among the resulting stability concepts are illustrated in \Cref{fig:stability-relations} \citep[see also][]{AzSa15a}.

For a given partition, 
several single-agent or group deviations might be possible.
Yet, some deviations seem to be more reasonable than others.
We say that a deviation is IR if 
the resulting partition is IR 
for all deviating agents. 
For all our considered stability concepts it holds that if an agent has a deviation (that is potentially not IR), then she also has an IR deviation where she forms a singleton coalition.

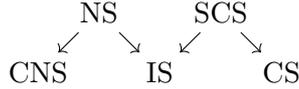
\begin{figure}
	\centering
	\begin{tikzpicture}
	
	\pgfmathsetmacro\lengthunit{0.8}
	\pgfmathsetmacro\heightunit{0.8}
	\node (NS) at (0,2*\heightunit) {NS};
	\node (IS) at (\lengthunit,\heightunit) {IS};
	\node (CNS) at (-\lengthunit,\heightunit) {CNS};
	\node (SC) at (2*\lengthunit,2*\heightunit) {SCS};
	\node (C) at (3*\lengthunit,\heightunit) {CS};
	
	\draw[->] (NS) edge (CNS);
	\draw[->] (NS) edge (IS);
	\draw[->] (SC) edge (IS);
	\draw[->] (SC) edge (C);
	
	\end{tikzpicture}
	\caption{Relations among our stability concepts.
		Arrows 
		indicate implications. 
		For example, 
		strict core stability (SCS) implies
		core stability (CS) and
		individual stability~(IS). 
		\label{fig:stability-relations}}
\end{figure}

\subsection{Dynamic Coalition Formation}
We now introduce our model of dynamic coalition formation over time, and the concepts of \emph{resent} and \emph{appreciation}.
Throughout the paper, we consider sequences of partitions $(\pi^t)_{t\ge 0}$, where for every $t\ge 1$, 
$\pi^t$ evolves from $\pi^{t-1}$ by means of some single-agent or group deviation.
We assume that both the initial coalition structure $\pi^0$ and the initial utility vectors $u_i^0$ for each agent $i\in N$ are given. 
However, utilities change over time as follows. 
Under \emph{resent}, agents 
decrease their utilities for all deviators that leave them (by one), while under \emph{appreciation}, agents increase their utilities 
for all deviators that join them (by one).\footnote{Note that our choice of decreasing, resp., increasing the utilities by one is somewhat arbitrary, as our theoretical results hold for any fixed increase or decrease of utilities. However, note that in case the utility change in each round is not constant, our convergence guarantees are no longer applicable, as, for instance, run-and-chase situations can occur.}
More formally, if for some $t\ge 1$, $\pi^t$ evolves from $\pi^{t-1}$ via a
single-agent deviation of agent $k\in N$, then, for $i,j\in N$,
\begin{itemize}
	\item for \emph{resentful} agents, $u_i^t(j)$ arises from $u_i^{t-1}(j)$ as
	\[u_i^t(j) = \begin{cases}
		u_i^{t-1}(j) -1 & i\neq k, j=k, j\in \pi^{t-1}(i),\\
		u_i^{t-1}(j) & \text{else.}\\
	\end{cases}\]
	\item for \emph{appreciative} agents, $u_i^t(j)$ arises from $u_i^{t-1}(j)$ as
	\[u_i^t(j) = \begin{cases}
		u_i^{t-1}(j) +1 & i\neq k, j=k, j\in \pi^{t}(i),\\
		u_i^{t-1}(j) & \text{else.}\\
	\end{cases}\]
\end{itemize}

If for $t\ge 1$, $\pi^t$ evolves from $\pi^{t-1}$ via a group deviation of $C\subseteq N$, then, for $i,j\in N$,
\begin{itemize}
	\item for \emph{resentful} agents, $u_i^t(j)$ arises from $u_i^{t-1}(j)$ as
	\[u_i^t(j) = \begin{cases}
		u_i^{t-1}(j) -1 & i\notin C, j\in C, j\in \pi^{t-1}(i),\\
		u_i^{t-1}(j) & \text{else.}\\
	\end{cases}\]
	\item for \emph{appreciative} agents, $u_i^t(j)$ arises from $u_i^{t-1}(j)$ as
	\[u_i^t(j) = \begin{cases}
		u_i^{t-1}(j) +1 & i\neq j, i\in C, j\in C,\\
		u_i^{t-1}(j) & \text{else.}\\
	\end{cases}\]
\end{itemize}

We are concerned about sequences of partitions that evolve by deviations with respect to the current utilities of the agents.
For any stability concept $\alpha \in \{\mathit{NS, IS, CNS, CS, SCS}\}$, a sequence of partitions $(\pi^t)_{t\ge 0}$ is called an execution of an $\alpha$ \emph{dynamics} if $\pi^t$ evolves from $\pi^{t-1}$ through an $\alpha$ deviation with respect to the utility functions $(u_i^{t-1})_{i\in N}$.
If all deviations are individually rational, we call the dynamics \emph{individually rational}, e.g., individually rational NS dynamics in the case of Nash stability.

An execution of an $\alpha$ dynamics \emph{converges} if it terminates after a finite number of $T$ steps in a partition $\pi^T$ that is stable with respect to 
$(u_i^T)_{i\in N}$ 
under the stability notion $\alpha$.
We say that the $\alpha$ dynamics \emph{converges} if every execution of the $\alpha$ dynamics converges for every initial utility profile and partition. By contrast, the dynamics \emph{cycles} if there exists an infinite execution of the dynamics (for some initial utilities and partition).
The central question of this paper is when dynamics converge for resentful or appreciative agents. 

It is convenient to use a compact notation for utilities. We write $u_i^{t,A_i}(\pi) = A_i(\pi(i),u_i^t)$ and $u_i^{t,A_i}(C) = A_i(C,u_i^t)$ for the utility of agent $i$ at time $t$ for a partition $\pi\in \Partitions$ or for coalition $C\in \mathcal N_i$, respectively. If the {\caf} $A_i$ is clear from context, we usually omit it as superscript.

Before our main analysis, we present a useful lemma that holds for arbitrary dynamics. 
The lemma can be applied to show that, from a certain point onwards, every deviation occurs infinitely often in an infinite execution of a dynamic.
The proof is straightforward and is included in the appendix.

\begin{restatable}{lemma}{infinite}
	\label{lem:infoccurence}
	Let $(\pi^t)_{t\ge 0}$ be an infinite sequence of partitions induced by single-agent (or group) deviations. Then, there exists a $t_0\ge 0$ such that every single-agent (or group) deviation performed at some time $t\ge t_0$ occurs infinitely often.
\end{restatable}

Lastly, we call an infinite sequence of partitions $\pi = (\pi^t)_{t\ge 0}$ \emph{periodic} if there exist $t_0\in \mathbb N$ and $p\in \mathbb N$ such that, for all $k\in \mathbb N_0$ and $l\in \{0,\dots, p-1\}$, it holds that $\pi^{t_0+kp+l} = \pi^{t_0+l}$.

\subsection{Properties of Aggregation Functions}
We now introduce some
useful properties
of \cafs.

For any $i\in N$, a {\caf} $A_i$ satisfies

\begin{itemize}
	\item \emph{aversion to enemies} (ATE) if, for all 
	coalitions $C\in \mathcal{N}_i$,
	agents $j\in C\setminus\{i\}$, 
	and utility vectors $u_i\in \mathbb Q^n$ with $u_i(j) < 0$, it holds that $A_i(C, u_i) \le A_i(C\setminus\{j\}, u_i)$. In other words, the aggregated utility is weakly better whenever an enemy leaves $i$'s coalition.
	\item \emph{individually rational aversion to enemies} (IR ATE) if,
	for all 
	coalitions $C\in \mathcal{N}_i$,
	agents $j\in C\setminus\{i\}$, 
	and utility vectors $u_i\in \mathbb Q^n$ with $u_i(j) < 0$, it holds that $A_i(C, u_i) \le A_i(C\setminus\{j\}, u_i)$
	if $A_i(C,u_i)\ge A_i(\{i\},u_i)$. In other words, the aggregated utility is weakly better when an enemy leaves one of $i$'s individually rational coalitions.
	\item \emph{enemy monotonicity} (EM) if, 
	for all 
	coalitions $C\in \mathcal{N}_i$,
	agents $j\in N$, and utility vectors $u_i, u_i'\in \mathbb Q^n$ with $u_i(k) = u_i'(k)$ for all $k \neq j$ and $u_i'(j) < u_i(j) < 0$, it holds that $A_i(C, u_i) \ge A_i(C, u_i')$. In other words, decreasing the utility for an enemy cannot improve a coalition value.
	\item \emph{enemy domination} (ED) if, 
	for all 
	utility vectors $u_i\in \mathbb Q^n$ and 
	agents $j\in N\setminus \{i\}$, there exists a constant $c(u_i,j)$ such that for all utility vectors $u_i' \in \mathbb Q^n$ with $u_i'(k) \le u_i(k)$ for all $k\in N$ and $u_i'(j)\le c(u_i,j)$, it holds for every $C\in \mathcal{N}_i$ with $j\in C$ 
	that $A_i(C, u_i') < A_i(\{i\},u_i')$. In other words, an {\caf} satisfies ED if in case $i$'s utility for some agent $j$ is sufficiently negative and $i$'s utility for every other agent is bounded, then no coalition containing $j$ is individually rational for $i$.
\end{itemize}

All of these axioms capture the treatment of enemies.
The first two axioms deal with situations where an enemy leaves the agent's coalition, where ATE is stronger than IR ATE.
On the other hand, EM and ED are variable utility conditions describing situations where the utility for an enemy decreases or some agent turns into a very bad enemy, respectively. 
Apart from the implication between ATE and IR ATE, there are no other logical relationships between any pair of axioms.

\begin{example}\label{ex:MFHGviolateATE}
	In this example, we consider a game $(N,u)$ for which the {\caf} $\mathit{MF}$ violates ATE. 
	Let $N = \{a,b,c\}$ and let the single-agent utilities be $u_a(b) = -1$, $u_a(c) = -3$, $u_b(a) = 1$, and $u_b(c) = -1$. (The utilities $u_c(a)$ and $u_c(b)$ are irrelevant.)
	
	\begin{figure}
		\centering
		\begin{tikzpicture}[
		element/.style={shape=circle,draw, fill=white}
		]
		\node[element] (a) at (-1.2,0) {$a$};
		\node[element] (b) at (60:1.2) {$b$};
		\node[element] (c) at (300:1.2) {$c$};

		\draw[bend right = 20, ->] (a) edge node[midway, fill = white] {$-1$} (b);
		\draw[->] (a) edge node[midway, fill = white] {$-3$} (c);
		\draw[bend right = 20, ->] (b) edge node[midway, fill = white] {$1$} (a);
		\draw[->] (b) edge node[midway, fill = white] {$-1$} (c);
		\end{tikzpicture}
		\caption{Illustration of teh game $(N,u)$ in \Cref{ex:MFHGviolateATE}. An edge with weight $w$ from agent $x$ to agent $y$ means that $u_x(y) = w$.}
	\end{figure}
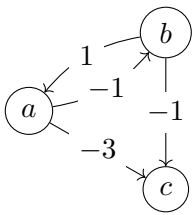

	Then, removing an enemy can make an agent worse. Indeed, $\mathit{MF}_a(N,u_a) = -2 > -3 = \mathit{MF}_a(\{a,c\},u_a)$.
	Hence, $\mathit{MF}$ violates ATE. 
	 On the other hand, as we will see in \Cref{prop:axioms}, removing an enemy from an individually rational coalition cannot decrease the utility in an MFHG. For instance, $\mathit{MF}_b(N,u_b) = 0 < 1 = \mathit{MF}_b(\{a,b\},u_b)$. \hfill$\lhd$
\end{example}

Still, classical aggregation functions usually satisfy (most of) our introduced axioms. 
\begin{restatable}{proposition}{axioms}
	\label{prop:axioms}
	The additively separable {\caf} $\mathit{AS}_i$ satisfies ATE, IR ATE, EM, and ED.
	The modified fractional {\caf}~$\mathit{MF}_i$ satisfies IR ATE, EM, and ED but violates ATE.
\end{restatable}
\begin{proof}
	We first consider additively separable aggregation and each axiom separately.
	
	\begin{itemize}
		\item ATE holds because a sum gets smaller when removing a negative summand.
		\item IR ATE follows from ATE.
		\item EM holds because a sum can only get smaller when decreasing every summand (and decreasing one of them strictly).
		\item ED holds because $\sum_{k\in C\setminus\{i\}}u_i(k)$ gets negative if all summands only diminish and, for $j\in C\setminus\{i\}$, $u(j) \le - 1 - \sum_{k\in C\setminus \{j\}}u_i(k)$. Hence, we can choose the constant $c(u_i,j) =  \min_{C\in \mathcal N_i\colon j\in C} - 1 - \sum_{k\in C\setminus \{j\}}u_i(k)$.
	\end{itemize}
	
	Next, we consider modified fractional aggregation. The violation of ATE is considered in \Cref{ex:MFHGviolateATE}.
	\begin{itemize}
		\item IR ATE: Let $C\subseteq \mathcal N_i$, $j\in C\setminus\{i\}$, and $u_i\in \mathbb Q^n$ with $u_i(j) < 0$. Moreover, suppose that $\mathit{MF}_i(C,u_i)\ge 0$. Then, $|C| \ge 3$ because otherwise $\mathit{MF}_i(C,u_i) = u_i(j) < 0$. Consequently, 
		\begin{align*}
		&\mathit{MF}_i(C\setminus\{j\},u_i) \ge \mathit{MF}_i(C,u_i) \\
		\Leftrightarrow & \frac{\sum_{k\in C\setminus \{i,j\}}u_i(k)}{|C|-2} \ge \frac{\sum_{k\in C\setminus\{i\}}u_i(k)}{|C|-1}\\
		\Leftrightarrow & (|C| - 1) {\sum_{k\in C\setminus \{i,j\}}u_i(k)} \ge (|C|-2)\sum_{k\in C\setminus\{i\}}u_i(k)\\
		\Leftrightarrow & \sum_{k\in C\setminus \{i,j\}}u_i(k) \ge (|C|-2)u_i(j)\\
		\Leftrightarrow & \frac{\sum_{k\in C\setminus\{i\}}u_i(k)}{|C|-1} \ge u_i(j).
		\end{align*}
		Hence, whenever $u_i(j) < 0$ and $\mathit{MF}_i(C,u_i)\ge 0$, the utility increases in the coalition where $j$ has left.	
		\item EM and ED hold for the same reasons as for ASHGs because the division by a positive number does not change the respective arguments.
	\end{itemize}
\end{proof}

\section{Dynamics for Resentful Agents} \label{sec:resent}
In this section, we study the convergence of different types of dynamics for resentful agents.
We start by considering (S)CS and IS dynamics, before turning to CNS and NS dynamics.

\subsection{Core Stability and Individual Stability}
If deviating agents need consensus from their new coalition, it turns out that resent is a strong force to establish convergence.
The intuitive reason for this is that an agent $a$ can only leave an agent $b$ for a limited number of times until resent prevents that they form a joint coalition again. In fact, otherwise $b$'s utility for $a$ becomes arbitrarily negative and $b$ no longer gives $a$ her consent to join. 
We will prove that SCS dynamics, and thereby also CS and IS dynamics, always converge for a wide class of {\cafs}. 
\begin{restatable}{theorem}{resentsc}
	\label{thm:resent:SC-converges}\label{cor:resent-convergence}
	The \mbox{SCS}, CS, and IS dynamics converge for resentful agents whose {\cafs} satisfy
	aversion to enemies and enemy monotonicity.
\end{restatable}
\begin{proof}
	It suffices to consider SCS dynamics because every CS dynamics and every IS dynamics is also an SCS dynamics.
	
	Let a hedonic game $(N,u^0)$ with resentful agents be given 
	where every agent $i\in N$ has a {\caf}~$A_i$ that satisfies
	aversion to enemies and enemy monotonicity.
	The key insight to show convergence of the SCS dynamics is to prove that in every infinite sequence of SCS deviations, 
	it happens infinitely often that
	the \emph{non-negative} single-agent valuation of some agent for another agent is decreased (due to resent). This is a contradiction, as the number of such deviations is bounded by $\sum_{i,j\in N\colon u_i^0(j)\ge 0} (\lfloor u_i^0(j)\rfloor + 1)$.
	
	Consider an infinite execution $(\pi^t)_{t\ge 0}$ of the SCS dynamics. 
	For the sake of contradiction, assume that there is a step $t_0\ge 0$ in this execution such that, starting with $t_0$, it never happens again that the non-negative valuation of some agent for another agent is decreased. The first step towards a contradiction is to show that every SCS deviation is a Pareto improvement in this situation.\footnote{An outcome is a \emph{Pareto improvement} over another outcome if it is weakly better for all agents and strictly better for some agents.} Indeed, every deviation is weakly improving for every agent in the coalitions that are abandoned: As no agent decreases her non-negative utility for another agent, each abandoned agent is solely abandoned by agents for which she momentarily has a negative utility. 
	We can then iteratively apply aversion to enemies to conclude that the deviation (weakly) increases 
	the utility of each abandoned agent.
	Further, since by the definition of an SCS deviation, every agent of the newly formed coalition is weakly improving, and some of them are strictly improving, each SCS deviation is a Pareto improvement. 
	
	Still, it is not clear, why the changes of the utility functions due to resent cannot cause sequences of Pareto improvements of infinite length, and we need enemy monotonicity to prove that this can never be the case. 
	To this end, we will show that every agent~$i$ can deviate at most once to each coalition $C\in \mathcal N_i$ while \emph{improving her utility} by doing so. This provides an upper bound for the length of any sequence of Pareto improving deviations.
	
	Consider a time $t \ge t_0$ and define for some agent $i \in N$ the set $\mathcal C_i^t = \{C\in \mathcal N_i \colon u_i^t(C) > u_i^t(\pi^t)\}$, i.e., the set of coalitions that would lead to an improvement of agent $i$ at time $t$. Consider the potential function $V^t = \sum_{i\in N} |\mathcal C_i^t|$. Note that $V^{t_0}\le \sum_{i\in N}|\mathcal N_i|$ is initially bounded and that the potential attains only non-negative integer values. We will show that the potential decreases in every time step after $t_0$.
	
	To this end, let $i\in N$ be some fixed agent and $t\ge t_0$ some fixed time. If $\pi^{t+1}(i) = \pi^t(i)$, then agent $i$ does not change her utilities, and $\mathcal C_i^{t+1} = \mathcal C_i^t$. 
	Further, if $\pi^{t+1}(i) \neq \pi^t(i)$ but agent~$i$ was not part of the deviating coalition, then, by our above assumption that no non-negative utility gets decreased, $u_i^t(j) < 0$ for all agents $j\in \pi^t(i)\setminus \pi^{t+1}(i)$. Hence, we can repeatedly apply aversion to enemies for all agents in $\pi^t(i)\setminus \pi^{t+1}(i)$ to conclude that $u_i^{t+1}(\pi^{t+1}) = u_i^{t}(\pi^{t+1}) \ge u_i^t(\pi^t)$.
	Further, as only $i$'s valuations of agents from $\pi^t(i)\setminus \pi^{t+1}(i)$ got modified (decreased), the only coalitions $C\in \mathcal N_i$ that might have changed their value for agent $i$ contain some agent in $\pi^t(i)\setminus \pi^{t+1}(i)$. Hence, we can (repeatedly) apply enemy monotonicity to conclude that $u_i^{t+1}(C)\le u_i^t(C)$ for all such coalitions.
	Hence, $\mathcal C_i^{t+1} \subseteq \mathcal C_i^t$. 
	
	Finally, if agent $i$ is part of the deviating coalition at step $t$, then she does not change her utilities for any agent and weakly improves the valuation of her coalition. Hence, $\mathcal C_i^{t+1} \subseteq \mathcal C_i^t$ and we can conclude that $V^{t+1}\le V^t$. 
	Further, if $i$ strictly improves her utility, then $u_i^{t+1}(\pi^{t+1}) > u_i^{t+1}(\pi^t)$ and it follows that $\pi^{t+1}(i)\in \mathcal C_i^t \setminus \mathcal C_i^{t+1}$. Hence, $\mathcal C_i^{t+1} \subsetneq \mathcal C_i^t$. Since this has to be the case for at least one agent in every time step, we conclude that $V^{t+1} < V^t$.
\end{proof}

As the cardinal aggregation function $\mathit{AS}$ satisfies aversion to enemies and enemy monotonicity (cf. \Cref{prop:axioms}), \Cref{thm:resent:SC-converges} in particular implies that the  SCS, CS, and IS dynamics always converge in ASHGs for resentful agents.

Notably, \Cref{cor:resent-convergence} breaks down if we consider a {\caf} violating aversion to enemies, even if enemy monotonicity is still satisfied. 
Indeed, we can then ``ignore'' individual utilities. 
For instance, anonymous hedonic games where agents only care about the size of their coalitions satisfy enemy monotonicity.
In such games, resent is clearly irrelevant and there exist anonymous hedonic games where IS dynamics cycle \citep{BBW21a}.
Consequently, a result similar to \Cref{thm:resent:SC-converges} for aggregation functions that only satisfy enemy monotonicity cannot be obtained.
On the other hand, it remains an open question whether enemy monotonicity is necessary for \Cref{thm:resent:SC-converges}. 

Unfortunately, $\mathit{MF}$ violates aversion to enemies (\Cref{prop:axioms}), 
implying that \Cref{cor:resent-convergence} cannot be directly applied to MFHGs for resentful agents.
Nevertheless, if we require the performed SCS deviations to be individually rational, then we can achieve convergence for a class of games containing MFHGs (cf. \Cref{prop:axioms}).

\begin{restatable}{theorem}{resentscenemy}
	\label{thm:resent:SC-can-converge}\label{cor:resent-can-converge}
	The individually rational SCS, CS, and IS dynamics converge for resentful agents whose {\cafs} satisfy
	individually rational aversion to enemies and enemy monotonicity.
\end{restatable}
\begin{proof}
	Again, it suffices to consider SCS dynamics. 
	
	Let a hedonic game $(N,u^0)$ with resentful agents be given
	where every agent $i\in N$ has a {\caf}~$A_i$ that satisfies individually rational aversion to enemies and enemy monotonicity. Assume that agents perform only individually rational deviations.
	Assume for contradiction that there exists an infinite execution of the individually rational SCS dynamics. Since the single-agent utilities for other agents can only decrease and are bounded by the initial partition of the sequence, there exists a time $t_0\ge 0$ in the infinite dynamics after which no agent can be left by an agent for which she has non-negative utility. 
	By \Cref{lem:infoccurence}, there exists a time $t_1\ge t_0$ such that every SCS deviation performed after time $t_1$ is performed infinitely often. 
	
	Now, consider the first group deviation performed after time $t_1$, where a coalition $C$ is formed and let $i\in C$ be an arbitrary agent in this coalition. Note that the formation of $C$ was an individually rational deviation, and therefore $C$ is individually rational for agent $i$ at time $t_1$. 
	We claim that $i$'s utility cannot decrease until the next time when the same deviation is performed, and that all coalitions that $i$ is part of until then are IR. Until the next repetition of the same deviation, there are two potential cases when $i$'s utility is affected. First, it can happen that $i$ is part of a deviating coalition. 
	Clearly, this cannot decrease her utility, and therefore not affect her individual rationality. Second, it can happen that $i$'s coalition is left by a set of agents $D$. 
	Since~$i$ is resentful and decreases her utility for all agents in $D$, it must be the case that her utility is negative for every agent in $D$ at the point in time where the deviation involving $D$ occurs. 
	
	Moreover, $i$'s coalition is IR when $D$ leaves.
	Indeed, assume for contradiction that $i$'s coalition is not IR, and consider the first time after the formation of $C$ when~$i$ is in a coalition that is not IR. By our analysis before, this can only happen after $i$ performed a group deviation, but since this deviation originated from an IR coalition, the resulting coalition must also be IR, a contradiction. 
	
	Now, by applying individually rational aversion to enemies for each agent in $D$ one after another, we can conclude that $i$ cannot have decreased her utility at the point in time where she is left by $D$. 
	
	Hence, at the next time, when coalition $C$ is formed, by the same deviation as our initial deviation, none of the agents in $C$ has decreased her utility. 
	Moreover, enemy monotonicity implies that the utility of $C$ of any agent for $C$
	can only have decreased since the last time when $C$ was formed. 
	Hence, no agent can strictly improve her utility when forming $C$ again, a contradiction. 
	We conclude that the dynamics cannot run infinitely. 
\end{proof}

It remains open whether general SCS, CS, or IS dynamics for resentful agents  may cycle in an MFHG.

\subsection{Contractual Nash Stability and Nash Stability}

For individually rational NS dynamics, 
resent helps to establish convergence for a wide class of games.
\begin{restatable}{theorem}{resentNSIR}
	\label{thm:resent:NS-IR-converges} \label{cor:resent-special-classes}
	The individually rational NS dynamics converges for resentful agents whose {\cafs} satisfy enemy domination.
\end{restatable}

\begin{proof}
	Let a hedonic game $(N,u^0)$ with resentful agents be given
	where every agent $i\in N$ has a {\caf}~$A_i$ that satisfies enemy domination. Assume for contradiction that there is an infinite execution of the individually rational NS dynamics $(\pi^t)_{t\ge 0}$. Suppose that, for every $t\ge 1$, $\pi^t$ evolves from $\pi^{t-1}$ by an individually rational NS deviation of agent~$d^t$.
	
	By \Cref{lem:infoccurence}, there exists $t_0\ge 0$ such that every deviation performed after $t_0$ is performed infinitely often. We will reach a contradiction in two steps. First, we use enemy domination to show that no agent can ever be abandoned by an agent that she joined after $t_0$. Then, as a second step, we use this insight to show the existence of a non-negative potential function decreasing in every time step after $t_0$.

	For the first step, let $i\in N$ be an agent and let $C\in \mathcal N_i$ be a coalition such that agent~$i$ performs a deviation at time $t_1\ge t_0$ to form coalition $C$. 
	Then, no agent from $C\setminus \{i\}$ can abandon agent~$i$ in any of their deviations after time~$t_0$. 
	Assume for contradiction that there exists an agent $j\in C\setminus \{i\}$ who abandons agent~$i$ after time $t_0$. 
	Since the aggregation function satisfies enemy domination, there is a constant $c(u_i,j)$ such that for all utility vectors $u'_i$ with $u'_i(k) \le u^{t_1}_i(k)$ for all $k\in N$ and $u'_i(j)\le c(u_i,j)$, it holds that $u_i'(C) < u_i'(\{i\})$.
	
	Since every deviation occurs infinitely often, there exists a time $t_2\ge t_1$ such that agent~$j$ has abandoned agent~$i$ for at least $u^{t_1}_i(j) - c(u_i,j)$ times between time $t_1$ and time $t_2$. 
	Since agent~$i$ is resentful, this implies that $u^{t_2}_i(j) \le u^{t_1}_i(j) - (u^{t_1}_i(j) - c(u_i,j)) = c(u_i,j)$. Additionally, resentful agents can only decrease utilities for other agents. 
	Therefore, it holds for all $t\ge t_2$ and $k\in N$ that $u^t_i(k) \le u^{t_1}_i(k)$ and $u^t_i(j)\le c(u_i,j)$. 
	Consequently, it follows from enemy domination for all times $t\ge t_2$ that agent~$i$ cannot deviate to form coalition~$C$ again because this deviation would not be individually rational. 
	However, this contradicts the fact that every deviation after time~$t_1$ has to be performed infinitely often. This establishes the second step, i.e., that an agent can never be left by an agent that is part of a coalition which she joins after time $t_0$. 
	
	For the second step, we will now define a potential function that is bounded from below by $0$, integer-valued, and strictly decreasing in every step $t\ge t_1$. 
	Therefore, fix an agent~$i\in N$ and define $A_i = \{j\in N\colon j\textnormal{ abandons } i \textnormal{ after time }t_0\}$.
	Given $t\ge t_0$, we have to distinguish two cases. If $i$ already performed a deviation after time $t_0$, then let $d_i(t) = \max\{t_0\le t'\le t\colon d^{t'} = i\}$ be the last time that $i$ performed a deviation between $t_0$ and $t$. 
	In this case, define $P_i(t) = \{C\subseteq N\setminus A_i \colon i \in C, u_i^t(C) > u_i^t(\pi^{d_i(t)})\}$. 
	Otherwise, set $P_i(t) = 2^{N\setminus A_i}$. Note that $P_i(t)$ contains all coalitions which $i$ could \emph{potentially} form through a deviation after time $t$. Define the potential $\Lambda(t) = \sum_{j\in N}|P_j(t)|$.
	
	We observe that $P_i(t+1) = P_i(t)$ if $i\neq d^{t+1}$. 
	Assume now that $i = d^{t+1}$. 
	If the first deviation of $i$ occurs at time~$t$, then $P_i(t+1)\subseteq P_i(t)\setminus\{\pi^t(i)\}$.
	Otherwise, it must hold that $\pi^t(i) = \pi^{d_i(t)}(i)$, i.e., $i$ is part of the same coalition in step $t$ to which it deviated in step $d_i(t)$, which is the last deviation performed by $i$. 
	Indeed, due to the second step, no agent that was present when $i$ joined $\pi^{d_i(t)}(i)\setminus\{i\}$ can have left and~$i$ is not allowed to leave if any agent from $N\setminus \pi^{d_i(t)}(i)$ is still present. 
	Also, agent $i$'s utilities for agents in $N\setminus A_i$ have not changed since her last deviation and therefore $u^t_i(j) = u^{d_i(t)}(j)$ for all $j\in N\setminus A_i$. Thus, $i$ derives the same utility from $\pi^t(i)$ and $\pi^{d_i(t)}(i)$. 
	It follows that $P_i(t+1)\subseteq P_i(t)\setminus\{\pi^t(i)\}$. 
	Consequently, in each case, $\Lambda(t + 1) < \Lambda(t)$. As $\Lambda(t)\ge 0$ for all $t\ge t_1$, the dynamics can run for at most $\Lambda(t_1)$ 
	steps after time $t_1$. This completes the proof.	
\end{proof}

As $\mathit{AS}$ and $\mathit{MF}$ satisfy enemy domination (\Cref{prop:axioms}), \Cref{thm:resent:NS-IR-converges} implies that the individually rational NS dynamics converges in ASHGs and MFHGs for resentful agents. 
However, we do not know under which conditions resent is sufficient to guarantee convergence for arbitrary (not necessarily individually rational) NS dynamics.
In this case, our proof for \Cref{thm:resent:NS-IR-converges} no longer works because it is possible that agents join coalitions for which they have an arbitrarily low utility (if the utility for their abandoned coalition was even worse).
In fact, slightly counterintuitive, there is a non-trivial example of a cycling NS dynamics in an MFHG for resentful agents.
\begin{restatable}{theorem}{mfhgresent}
	\label{th:mfhg-resent}
	The NS dynamics may cycle in MFHGs for resentful agents.
\end{restatable}
\begin{proof}
	We now describe an involved example of an MFHG together with an infinite periodic sequence of NS deviations for resentful agents. 
	We construct this example in a way such that each agent leaves every other agent exactly once in each cycle. 
	This establishes that the agents' preference between relevant coalitions is \emph{maintained}: if a deviating agent prefers a joined coalition $C_1$ to an abandoned coalition $C_2$
	before (and during) the \emph{first} execution of the cycle, then it still prefers $C_1$ to $C_2$ before (and during) \emph{each} execution of the cycle.

	\begin{table}
		\caption{Example from \Cref{th:mfhg-resent}. Utilities after $x$ cycles of deviations. Each row depicts the utility that one agent has for the other agents.} \label{table:mfhg-resent}
		\centering
		\begin{tabular}{ c|c|c|c|c|c|c } 
			& $a$ & $a'$ & $b$ & $b'$ & $c$ & $c'$ \\ \hline
			$a$ & $-$ & $20-x$ & $10-x$ & $230-x$ & $0-x$ & $230-x$ \\ \hline
			$a'$ & $110-x$ & $-$ & $30-x$ & $120-x$ & $30-x$ & $100-x$ \\ \hline
			$b$ & $0-x$ & $230-x$ & $-$ & $20-x$ & $10-x$ & $230-x$ \\ \hline
			$b'$ & $30-x$ & $100-x$ & $110-x$ & $-$ & $30-x$ & $120-x$ \\ \hline
			$c$ & $10-x$ & $230-x$ & $0-x$ & $230-x$ & $-$ & $20-x$ \\ \hline
			$c'$ & $30-x$ & $120-x$ & $30-x$ & $100-x$ & $110-x$ & $-$
		\end{tabular}
	\end{table}
	
	Consider the game with agent set $N = \{a,a',b,b',c,c'\}$ and utilities as depicted in \Cref{table:mfhg-resent}. The initial utilities result from setting $x = 0$ in \Cref{table:mfhg-resent}. 
	Note that the utility values are not chosen to be minimal but simply in a way that it can easily be verified that deviations are indeed NS deviations.
	We now present an infinite sequence $(\pi^t)_{t\ge 0}$ of partitions, always consisting of three coalitions. 
	For each partition, we refer to the first listed coalition as $C_1$, to the second as $C_2$, and the third as $C_3$. 
	For the sake of clarity, for each partition, we also specify which agent deviates to which coalition in the next step. 
	Specifically, for $n\geq 0$, we have
	\begin{itemize}
		\item $\pi^{18n+1}=\{\{b',a'\},\{a\},\{c',c,b\}\}$ with agent $b$ deviating to $C_1$, 
		\item $\pi^{18n+2}=\{\{b',a',b\},\{a\},\{c',c\}\}$ with agent $c$ deviating to $C_1$, 
		\item $\pi^{18n+3}=\{\{b',a',b,c\},\{a\},\{c'\}\}$ with agent $a'$ deviating to $C_3$, 
		\item $\pi^{18n+4}=\{\{b',b,c\},\{a\},\{c',a'\}\}$ with agent $a'$ deviating to $C_2$, 
		\item $\pi^{18n+5}=\{\{b',b,c\},\{a,a'\},\{c'\}\}$ with agent $c$ deviating to $C_2$,
		\item $\pi^{18n+6}=\{\{b',b\},\{a,a',c\},\{c'\}\}$ with agent $b'$ deviating to $C_3$,
		\item $\pi^{18n+7}=\{\{b\},\{a,a',c\},\{c',b'\}\}$ with agent $c$ deviating to $C_3$,
		\item $\pi^{18n+8}=\{\{b\},\{a,a'\},\{c',b',c\}\}$ with agent $a$ deviating to $C_3$,
		\item $\pi^{18n+9}=\{\{b\},\{a'\},\{c',b',c,a\}\}$ with agent $b'$ deviating to $C_2$,
		\item $\pi^{18n+10}=\{\{b\},\{a',b'\},\{c',c,a\}\}$ with agent $b'$ deviating to $C_1$,
		\item $\pi^{18n+11}=\{\{b,b'\},\{a'\},\{c',c,a\}\}$ with agent $a$ deviating to $C_1$,
		\item $\pi^{18n+12}=\{\{b,b',a\},\{a'\},\{c',c\}\}$ with agent $c'$ deviating to $C_2$,
		\item $\pi^{18n+13}=\{\{b,b',a\},\{a',c'\},\{c\}\}$ with agent $a$ deviating to $C_2$,
		\item $\pi^{18n+14}=\{\{b,b'\},\{a',c',a\},\{c\}\}$ with agent $b$ deviating to $C_2$,
		\item $\pi^{18n+15}=\{\{b'\},\{a',c',a,b\},\{c\}\}$ with agent $c'$ deviating to $C_1$,
		\item $\pi^{18n+16}=\{\{b',c'\},\{a',a,b\},\{c\}\}$ with agent $c'$ deviating to $C_3$,
		\item $\pi^{18n+17}=\{\{b'\},\{a',a,b\},\{c,c'\}\}$ with agent $b$ deviating to $C_3$, and
		\item $\pi^{18n+18}=\{\{b'\},\{a',a\},\{c,c',b\}\}$ with agent $a'$ deviating to $C_1$.
	\end{itemize}

	Then, it is possible to verify that for $k\ge 1$, $\pi^{k-1}$ leads to $\pi^k$ by means of an NS deviation. Hence, we have presented an MFHG with an infinite sequence of NS deviations for resentful agents. 
\end{proof}

This result indicates that some condition like aversion to enemies is probably needed for establishing a convergence guarantee for general NS dynamics; however, it remains open whether such a result is possible  (even for ASHGs). 
Notably, this question for {\cafs} satisfying aversion to enemies is the same as asking whether a CNS dynamics may cycle: For resentful agents in case of a cycling NS dynamics, there  is also a cycling CNS dynamics.
\begin{restatable}{proposition}{NSvCNS}
	\label{prop:resent-CNS-vs-NS}
	For resentful agents with {\cafs} satisfying aversion to enemies, every sequence of NS deviations contains only finitely many deviations that are not CNS deviations.
\end{restatable}
\begin{proof}
	Let a hedonic game $(N,u^0)$ with resentful agents be given
	where every agent $i\in N$ has a {\caf}~$A_i$ that satisfies aversion to enemies. 
	Furthermore,
	assume that there exists an infinite sequence $(\pi^t)_{t\ge 0}$ of partitions resulting from NS deviations of resentful agents.
	By \Cref{lem:infoccurence}, there exists a time $t_0\ge 0$ such that every deviation performed after $t_0$ must occur infinitely often.
	
	Define $\mathcal L = \{(i,j)\in N^2 \colon i \textnormal{ left by } j \textnormal{ after time }t_0\}$, i.e., the set of pairs of such agents. 
	Let $(i,j)\in \mathcal L$. 
	Then, there exists a time $t(i,j)\ge t_0$ such that $i$ was left by $j$ for at least $\lfloor u^{t_0}_i(j)\rfloor + 1$ times after time $t_0$ and before time $t(i,j)$. 
	Consider the time $t_1 = \max\{t(i,j)\colon (i,j)\in\mathcal L\}$. We claim that all deviations after time $t_1$ are CNS deviations. Indeed, assume that agent~$i$ is left by agent~$j$ at time $t\ge t_1$. By construction, as $t(i,j)\le t_1$, it holds that $u^t_i(j) < 0$. 
	Consequently, enemy monotonicity implies that $A_i(\pi^t,u_i^t)\le A_i(\pi^{t+1},u_i^t)$. Hence, every deviation after time $t_1$ is a CNS deviation.
	Thus, there are only finitely many deviations
	(at most $t_1$ many) that are not CNS deviations.
\end{proof}

\section{Dynamics for Appreciative Agents} \label{sec:apprec}
We now turn to analyzing the effects of appreciation on the convergence of different types of dynamics. 
Here, as statements for general {\cafs} would require the introduction of (even) further axioms, we focus on $\mathit{AS}$ and $\mathit{MF}$ instead. 
We start by establishing a close connection between cycling dynamics for resentful and appreciative agents in ASHGs, highlighting a close connection between the two studied models.
Subsequently, we analyze CS and (C)NS dynamics.

\subsection{From Resent to Appreciation}
We describe how we can transform certain types of infinite sequences of deviations for resentful agents to sequences for appreciative agents and vice versa. 
We focus on ASHGs, yet believe that similar statements can hold for other classes of hedonic games. We start with Nash stability.

\begin{restatable}{theorem}{resentappreciation}
	\label{thm:NS-resent-vs-appreciation-ASHG}
	The following statements are equivalent:
	\begin{enumerate}
		\item There exists an ASHG admitting an infinite and periodic sequence of NS deviations for resentful agents.
		\item There exists an ASHG admitting an infinite and periodic sequence of NS deviations for appreciative agents.
	\end{enumerate}
\end{restatable}
The idea to prove \Cref{thm:NS-resent-vs-appreciation-ASHG} is to reverse a periodic fragment of an infinite sequence and to appropriately adjust the initial utilities. 
This essentially reverses the roles of resent and appreciation, as the agents that an agent $a$ leaves in the sequence for resentful agents correspond to the agents $a$ joins in the sequence for appreciative agents.
The formal proof is quite technical and we defer it to the appendix.

By \Cref{prop:resent-CNS-vs-NS},  \Cref{thm:NS-resent-vs-appreciation-ASHG} 
can be extended to also include infinite and periodic sequences of CNS deviations for resentful agents. 
In fact, the equivalence can be extended even further, 
as we show that in ASHGs with appreciative agents, 
the question whether there is a cycling NS dynamics is 
equivalent to asking for a cycling IS dynamics. 
The proof idea for the next statement is that in every infinite sequence of NS deviations,
there exists a certain time step from which on 
agent $a$ has a positive utility for each agent $b$ that joins $a$ 
(because $b$ has already joined $a$ sufficiently often). 
\begin{restatable}{proposition}{ISvNS}
	\label{obs:appreciation-IS-vs-NS}
	For appreciative agents in ASHGs
	every sequence of NS deviations contains only finitely many deviations that are not IS deviations.
\end{restatable}
\begin{proof}
	Let an ASHG $(N,u^0)$ with appreciative agents be given. 
	Furthermore,
	assume that there exists an infinite sequence $(\pi^t)_{t\ge 0}$ of partitions resulting from NS deviations. 
	By \Cref{lem:infoccurence}, there exists a time $t_0$ such that every deviation performed after $t_0$ occurs infinitely often.
	
	Define $\mathcal L = \{(i,j)\in N^2 \colon i \textnormal{ is joined by } j \textnormal{ after time }t_0\}$. 
	Let $(i,j)\in \mathcal L$. 
	Then, there exists a time $t(i,j)\ge t_0$ such that $i$ was joined by $j$ for at least $\lceil |u^{t_0}_i(j)|\rceil + 1$ times after time $t_0$ and before time $t(i,j)$. 
	Consider the time $t_1 = \max\{t(i,j)\colon (i,j)\in\mathcal L\}$. We claim that all deviations after time $t_1$ are IS deviations. Indeed, assume that agent~$i$ is joined by agent~$j$ at time $t\ge t_1$. By construction, as $t(i,j)\le t_1$, it holds that $u^t_i(j) > 0$. 
	Consequently, agent $i$ prefers  $\pi^{t+1}$ to $\pi^t$. 
	Hence, every deviation after time $t_1$ is an IS deviation.
	Thus, there are only finitely many deviations
	(at most $t_1$ many) that are not IS deviations.
\end{proof}

To sum up, combining \Cref{thm:NS-resent-vs-appreciation-ASHG} and \Cref{obs:appreciation-IS-vs-NS,prop:resent-CNS-vs-NS}, we get the following equivalences.
\begin{corollary}\label{co:equivalence}
	The following statements are equivalent:
	\begin{enumerate}
		\item There exists an ASHG admitting an infinite and periodic sequence of CNS deviations for resentful agents.
		\item There exists an ASHG admitting an infinite and periodic sequence of NS deviations for resentful agents.
		\item There exists an ASHG admitting an infinite and periodic sequence of NS deviations for appreciative agents.
		\item There exists an ASHG admitting an infinite and periodic sequence of IS deviations for appreciative agents.
	\end{enumerate}	
\end{corollary}

\subsection{Convergence for Appreciative Agents}
We now give an overview under which circumstances appreciation is (not) sufficient to guarantee convergence in MFHGs and ASHGs.
In contrast to resent, appreciation is not sufficient to guarantee convergence of CS dynamics.\footnote{For ASHGs, 
		the next statement can be extended to an ASHG 
		where initial valuations are symmetric by slightly modifying the game 
		presented by \citet[Figure 2]{ABS11c}.}

\begin{restatable}{theorem}{appreciativecore}
	\label{th:appreciative-core-cylce}
	The individually rational CS dynamics may cycle in ASHGs and MFHGs for appreciative agents.

\end{restatable}
\begin{proof}
	Let $N=\{a,b,c\}$ be the set of agents and let the agents' initial utilities be as follows: 
	$$u^0_a(b)=u^0_b(c)=u^0_c(a)=4, \ \ u^0_a(c)=u^0_b(a)=u^0_c(b)=1.$$
	Let $\pi^0=\{\{a\},\{b\},\{c\}\}$ and for $t>0$, let  
	\[
	\pi^t =
	\begin{cases}
	\{\{a,c\},\{b\}\}, & \text{if } t\bmod 3 = 0  \\
	\{\{a,b\},\{c\}\}, & \text{if } t\bmod 3 = 1 \\
	\{\{a\},\{b,c\}\}, & \text{if } t\bmod 3 = 2.
	\end{cases}
	\]
	We claim that $(\pi^t)_{t\geq 0}$ is an infinite sequence, where for every $t\geq 1$, $\pi^t$ evolves from $\pi^{t-1}$ by a core deviation of coalition $C^t$. 
	Specifically, we have
	\[
	C^t =
	\begin{cases}
	\{a,c\}, & \text{if } t\bmod 3 = 0  \\
	\{a,b\}, & \text{if } t\bmod 3 = 1 \\
	\{b,c\}, & \text{if } t\bmod 3 = 2.
	\end{cases}
	\]
	Thus in each cycle of length three, each agent performs exactly one core deviation with any other agent. Thus, for $t\geq 0$ with $t\bmod 3 = 0$ it holds that 
	$$u^t_a(b)=u^t_b(c)=u^t_c(a)=\frac{2}{3}t+4$$ and $$u^t_a(c)=u^t_b(a)=u^t_c(b)=\frac{2}{3}t+1\text{.}$$
	Thus, for each $t>0$, it holds that $u^t_a(b)>u^t_a(c)$, $u^t_b(c)>u^t_b(a)$, and $u^t_c(a)>u^t_c(b)$, implying that each of the deviations $C^t$ for $t>0$ is a core deviation if each agent $i\in N$ aggregates utilities according to $\mathit{AS}_i$ or $\mathit{MF}_i$. 
\end{proof}

However, in the games considered in \Cref{th:appreciative-core-cylce}, 
there exists an execution of the CS dynamics that converges. 
This raises the (open) question whether a converging execution of the CS dynamics exists for every initial state in ASHGs and MFHGs for appreciative agents.

Lastly, we consider IS and (C)NS dynamics. 
In ASHGs for appreciative agents,
it remains open whether IS and NS dynamics may cycle. 
In fact, we have seen in \Cref{obs:appreciation-IS-vs-NS} that these two questions are equivalent and in \Cref{co:equivalence} that they are very closely related to our open questions concerning resentful agents. On the other hand, for CNS, appreciation is sufficient to guarantee convergence.
\begin{restatable}{theorem}{apprecCNS}
	\label{thm:apprec-CNS}
	The CNS dynamics converges in ASHGs for appreciative agents. 
\end{restatable}
\begin{proof}
	Let an ASHG $(N,u^0)$ be given and consider an execution $(\pi^t)_{t\ge 0}$ of the CNS dynamics. Assume for contradiction that the dynamics is infinite. By \Cref{lem:infoccurence}, there exists $t_0\ge 0$ such that every deviation performed at time $t\ge t_0$ is performed infinitely often.

	Now, consider an agent $i$ that joins a coalition $C$ containing agent $j$ at some time $t\ge t_0$. 
	We claim that it is impossible that $j$ ever joins a coalition containing agent~$i$ after time $t_0$. 
	For the contrary, assume that this happens, and therefore happens infinitely often. 
	Then, since agent~$i$ and agent~$j$ join each other infinitely often,  there exists a time $t'\ge t_0$ such that $u^{t}_j(i) > 0$ and $u^{t}_i(j) > 0$ for all $t\ge t'$. Some time after $t'$, the agents~$i$ and $j$ will be in a joint coalition, again. 
	Then, to perform the deviation again where $i$ joins $C$, the two agents have to be dissolved, i.e., one of them has to leave the other.
	However, since they by now have positive utility for each other, each of them would block the other agent from leaving according to the additively separable {\caf}. Hence, this cannot happen, a contradiction. 
	Consequently, agent~$j$ cannot join a coalition containing $i$ after time $t_0$. 
	
	Thus, the utilities of an agent $i$ for agents in coalitions which she joins by a deviation are not affected after time $t_0$ by appreciation. 
	Therefore, there exists a global constant $U$ such that the maximum utility of any agent obtained after any deviation is $U$. We derive a contradiction by showing that some agent has to abandon a coalition of unbounded utility.

	To this end, we prove the following claim.
	
	\begin{claim}
		There exists an agent $i$ that abandons a coalition $C$ containing an agent $j$ that joins $i$ at some point during the dynamics.
	\end{claim}
	
	\renewcommand\qedsymbol{$\vartriangleleft$}
	\begin{proof}
		Assume for contradiction that no such agent exists. To derive a contradiction, we will construct an infinite sequence of coalitions increasing in utility with respect to the utility at time $t_0$. This cannot happen, because the number of coalitions an agent can be part of (and therefore the number of different utility values that an agent can achieve at a fixed time) is bounded.
		
		We claim that there exists an agent $d\in N$ and a sequence of coalitions $(C_k)_{k\ge 0}$ such that for every $k\ge 1$, the following two conditions hold:
		
		\begin{enumerate}
			\item $d$'s utility with respect to time $t_0$ is strictly increasing, i.e., $u^{t_0}_d(C_k) > u^{t_0}_d(C_{k-1})$.
			\item The coalition $C_k$ is formed by a deviation of agent $d$.
		\end{enumerate}
		
		Consider the first agent $d$ performing a deviation after time $t_0$, where $d$ abandons coalition $C_0$ to form coalition $C_1$. 
		Then, $u^{t_0}_d(C_1) > u^{t_0}_d(C_0)$. Hence, we have found the first step of the sequence. 
		Now, assume that we have constructed a sequence $(C_k)_{k=0}^m$ which satisfies the two conditions. 
		We know that $C_m$ is formed by a deviation of agent~$d$. Consider the next time $\hat t$ where some agent $e\in C_m$ leaves the coalition $C$ containing $C_m$ (which must happen, because the dynamics is infinite). 
		If $e \neq d$, then we derive a contradiction to our initial assumption because $e$ then leaves a coalition containing an agent, namely agent~$d$, that joined her by a deviation. 
		Hence, $e = d$. If there exists an agent $f\in C\setminus C_m$, then we again derive a contradiction because $f$ must have joined $d$ at some point. 
		This implies that $C = C_m$. 
		Set $C_{m+1}$ to the coalition joined by agent~$d$. 
		Then, $C_{m+1}$ clearly fulfills the second condition. Also, $u^{t_0}_d(C_{m+1}) = u^{\hat t}_d(C_{m+1}) > u^{\hat t}_d(C) = u^{\hat t}_d(C_m) = u^{t_0}_d(C_m)$. 
		In the first and last equality, we use that appreciation does not affect the utilities for agents joined by agent $d$ after time $t_0$. Hence, also the first condition is fulfilled. As such a sequence of coalitions cannot exist, we derive a contradiction, and the claim must hold. 
	\end{proof}
	Now, consider an agent~$i$ that abandons a coalition $C$ containing an agent $j$ that joins $i$ at some point during the dynamics. 
	Since the deviation where agent~$j$ joins $i$ happens infinitely often, there exists a time $T\ge t_0$ such that $u^t_i(j)\ge U - \sum_{l\in C\setminus \{j\}}u^{t_0}_i(l)$ for all $t\ge T$. 
	Consider the next time $T'\ge T$ where~$i$ abandons $C$ to form some coalition $D$. 
	Then, $u^{T'}_i(D) \le U \le u^{T'}_i(j) + \sum_{l\in C\setminus \{j\}}u^{t_0}_i(l) \le u^{T'}_i(C)$. 
	In the last inequality, we use that the utility of $i$ for other agents can only have increased since time $t_0$ (because of appreciation). Hence, this deviation was not beneficial for~$i$. This is a final contradiction showing that an infinite dynamics cannot exist.
	\renewcommand\qedsymbol{$\square$}
\end{proof}

We proved in \Cref{th:mfhg-resent} that NS dynamics may cycle in MFHGs for resentful agents.
``Reversing'' this sequence and appropriately adjusting the initial utilities leads to a cycling  NS dynamics for appreciative agents. 
We defer the details to the appendix.

\begin{restatable}{theorem}{mfhgapprec}
	\label{th:mfhg-apprec}
	The individually rational NS dynamics may cycle in MFHGs for appreciative agents.
\end{restatable}
It remains open whether IS dynamics may cycle in MFHGs for appreciative agents.
Note that the arguments from \Cref{obs:appreciation-IS-vs-NS} for showing the ``equivalence'' for IS and NS dynamics under appreciation do not work for~MFHGs.

\section{Dynamics with Resentful Deviatiors}\label{sec_dev-resent}

Previously, we have assumed that a deviation of an agent $a$ changes the utility other agents have for $a$. 
In contrast, one could also consider what happens if a deviation of $a$ changes $a$'s utility for other agents.
Under deviator-resent, we assume that an agent decreases her utility for all agents she abandons.

Formally, if for some $t\ge 1$, $\pi^t$ evolves from $\pi^{t-1}$ via a
single-agent deviation of agent $k\in N$,
then for \emph{deviator-resentful} agents, for $i,j\in N$, $u_i^t(j)$ arises from $u_i^{t-1}(j)$ as
\[u_i^t(j) = \begin{cases}
u_i^{t-1}(j) -1 & i= k, j\in \pi^{t-1}(i)\setminus \{i\},\\
u_i^{t-1}(j) & \text{else.}\\
\end{cases}\]

Similarly,  for deviator-resentful agents, if for $t\ge 1$, $\pi^t$ evolves from $\pi^{t-1}$ via a group deviation of $C\subseteq N$, then, for $i,j\in N$, $u_i^t(j)$ arises from $u_i^{t-1}(j)$ as
\[u_i^t(j) = \begin{cases}
u_i^{t-1}(j) -1 & i\in C, j\in \pi^{t-1}(i)\setminus C,\\
u_i^{t-1}(j) & \text{else.}\\
\end{cases}\]

In addition to the axioms introduced in \Cref{se:prelims}, we also define two more axioms that are only relevant for the results in this section. 
These are weak conditions about friends.

\begin{itemize}
	\item \emph{friend necessity} (FN) if,
	for all 
	coalitions $C\in \mathcal{N}_i$
	and utility vectors $u_i\in \mathbb Q^n$, 
	$A_i(C,u_i) > 0$ implies that there exists $j\in C\setminus\{i\}$ with $u_i(j) > 0$. In other words, an agent can only have a positive utility for a coalition if it contains a friend.
	\item \emph{single friend desire} (SFD) if, 
	for all 
	coalitions $C\in\mathcal{N}_i$, 
	agents $j\in C\setminus\{i\}$, and utility vectors $u_i\in \mathbb Q^n$ such $u_i(j)>0$ and $u_i(k)\leq 0$ for all $k\in C\setminus \{j\}$, it holds that $A_i(C, u_i) > A_i(C\setminus \{j\}, u_i)$. In other words, a coalition to which exactly one friend belongs is strictly preferred to the same coalition without this friend.
\end{itemize}

We show that additively separable and modified fractional utility aggregation satisfies these axioms.

\begin{proposition}\label{prop:allaxioms}
	The additively separable {\caf} $\mathit{AS}_i$ and the modified fractional {\caf} $\mathit{MF}_i$ satisfy FN and SFD.
\end{proposition}
\begin{proof}
	FN holds because a sum (or a sum divided by a positive number) can only be positive if some summand is positive.
	SFD holds for the additively separable aggregation because a sum gets larger when adding a positive number.
	SFD holds for the modified fractional aggregation because a negative fraction gets larger when adding a positive number to the numerator while the denominator increases.
\end{proof}

\subsection{Contractual Nash Stability and Nash Stability}

An intuitive reason why deviator-resent can contribute to the convergence of dynamics is that, after agent $a$ abandons a coalition $C$, $a$'s utility for $C$ decreases and thus $a$ is less likely to join $C$ again. 
However, deviator-resent does not resolve the run-and-chase example, implying that NS dynamics may cycle for a wide variety of hedonic games with deviator-resentful agents.
Indeed, consider the dynamics between two agents $a$ and $b$ where initially $a$ has utility $1$ for $b$ and $b$ initially has utility $-1$ for $a$. 
Then, $a$ will still always join $b$ in one step, while $b$ leaves $a$ in the next step (thereby decreasing $b$'s utility for $a$ even further). 
\begin{observation}
	The NS dynamics may cycle for deviator-resentful agents whose {\cafs} satisfy friend necessity and single friend desire.
\end{observation}
By \Cref{prop:allaxioms} this implies that the NS dynamics may cycle in ASHGs and MFHGs.

Moreover, somewhat surprisingly, also CNS dynamics may still cycle in ASHGs and MFHGs with deviator-resentful agents, even if deviations are restricted to be 
individually rational. 
Notably, this is in a clear contrast to our previous results for resentful agents where individual rationality was for all types of dynamics sufficient to guarantee convergence. 
\begin{proposition}\label{prop:CNSaIRcycle}
	The individually rational CNS dynamics may cycle in MFHGs and ASHGs for deviator-resentful agents. 
\end{proposition}
\begin{proof}
	Consider a cardinal hedonic game with $N = \{a,b,c\}$ 
	where the initial single-agent utilities are given as $u_a(b) = u_b(c) = u_c(a) = 0$ and $u_b(a) = u_c(b) = u_a(c) = -1$. 
	For $k \ge 0$, let $\pi^{3k} = \{\{a,b\},\{c\}\}$, $\pi^{3k+1} = \{\{a\},\{b,c\}\}$, and  $\pi^{3k + 2} = \{\{a,c\},\{b\}\}$. 
	
	Then, $(\pi^k)_{k\ge 0}$ represents an individually rational CNS dynamics with respect to additively separable and modified fractional utility aggregation.  
\end{proof}
It is possible to modify the previous examples to start a dynamics from the singleton partition.

In contrast to this, as soon as we enforce that agents only deviate to a non-singleton coalition if they strictly prefer it to being in a singleton coalition, convergence can be guaranteed.

\begin{proposition} \label{prop:dev-resent:CNS-converges}
	The \mbox{CNS} dynamics converges for deviator-resentful agents whose {\cafs} satisfy friend necessity and single friend desire if agents only deviate to non-singleton coalitions if they strictly prefer them to being in a singleton.
\end{proposition}
\begin{proof}
	Let a hedonic game $(N,u^0)$ with deviator-resentful agents be given
	where every agent $i\in N$ has a {\caf}~$A_i$ that satisfies friend necessity and single friend desire.
	Assume for the sake of contradiction that there exists an infinite sequence $(\pi^t)_{t\ge 0}$ of partitions resulting from CNS deviations of deviator-resentful agents where agents only deviate to non-singleton coalitions if they strictly prefer them to being in a singleton.
	By \Cref{lem:infoccurence}, there exists a time $t_0\ge 0$ such that every deviation performed after $t_0$ must occur infinitely often.
	This implies in particular that there is a step $t'_0\ge t_0$ such that an agent $i$ has a negative utility for all agents that $i$ leaves at some point after $t'_0$. 
	
	Now, let us fix a time step $t'\geq t'_0$ where some agent $d$ deviates into a non-singleton coalition.
	As every deviation is repeated (infinitely often) after $t_0\leq t'$, there has to be some time step where $d$ deviates again.
	In particular, let $t''$ be the smallest $t$ with $t>t'$ where $d$ deviates again after time $t'$. 
	For some time step $t$, let $C_{t}:=\pi^{t}(d)$ be the coalition of $d$ after step $t$.  
	We now examine the coalitions $C_{t'},\dots, C_{t''-1}$, i.e., the coalitions $d$ is part of between her two deviations. 
	
	As we assume that an agent only deviates into a non-singleton coalition if she strictly prefers it to being in a singleton coalition, $d$ prefers $C_{t'}$ to being in a singleton. 
	Hence, by friend necessity, there is an agent $a\in C_{t'}$ for which $d$ has positive utility at time $t' - 1$.
	By single friend desire, it follows that there is an agent $a\in C_{t''-1}$ for which $d$ has a positive utility. Indeed, otherwise, there is some $t\in [t',t''-2]$ such that there is a friend of $d$ in $C_t$ but not in $C_{t+1}$. However, this cannot happen because, due to single friend desire, $d$ strictly prefers $C_t$ to $C_{t+1}$, implying that $d$ would have vetoed the deviation taking place in time step $t+1$. 
	This implies that there is at least one agent for which $d$ has positive utility in $C_{t''-1}$. 
	However, $d$ leaves this agent in time step $t''$. 
	This contradicts our initial assumption that each agent has a negative utility for all agents they leave in some step after $t'_0\leq t'\leq t''$.
\end{proof}

As the {\cafs} $\mathit{AS}_i$ and $\mathit{MF}_i$ satisfy friend necessity and single friend desire (cf. \Cref{prop:allaxioms}), \Cref{prop:dev-resent:CNS-converges} implies that, for every ASHG and MFHG with deviator-resentful agents, some execution of the CNS dynamics converges.

\subsection{Core Stability and Individual Stability}

We now turn to the consent-based stability concepts individual stability and (strict) core stability. 
Here, the influence of deviator-resent is more profound. 

\subsubsection{Individual Rationality}
We start by considering individually rational SCS dynamics. Here, like for resentful agents, convergence is guaranteed for a wide class of CAFs.
\begin{proposition}\label{thm:dev-resent:SCS-IR-converges}
	The individually rational SCS, CS, and IS dynamics converge for deviator-resentful agents whose CAFs satisfy enemy domination. 
\end{proposition}
\begin{proof}
	We show the statement for SCS deviations, which implies convergence of CS and IS dynamics.
		Let a hedonic game $(N,u^0)$ with deviator-resentful agents be given
	where every agent $i\in N$ has a {\caf}~$A_i$ that satisfies enemy domination. 
	Assume for the sake of contradiction that there exists an infinite sequence $(\pi^t)_{t\ge 0}$ of partitions resulting from SCS deviations of deviator-resentful agents.
	By \Cref{lem:infoccurence}, there exists a time $t_0\ge 0$ such that every deviation performed after $t_0$ must occur infinitely often.

	We first show that if two agent $i$ and $j$ perform a core deviation together at some point after $t_0$, then $i$ cannot leave $j$ nor can $j$ leave $i$ after $t_0$. 
	We prove that $i$ cannot leave $j$ (the other case is symmetric).
	By enemy domination, there exists a constant $c(u^{t_0}_i,j)$ such that for all utility vectors $u'_i \in \mathbb Q^n$ with $u'_i(k) \le u^{t_0}_i(k)$ for all $k\in N$ and $u'_i(j)\le c(u^{t_0}_i,j)$, it holds for every $C\in \mathcal{N}_i$ with $j\in C$  that $A_i(C, u'_i) < A_i(\{i\},u'_i)$.
	
	As each deviation after $t_0$ occurs infinitely often, the fact that $i$ leaves $j$ after $t_0$ infinitely often implies that there exists a time $t_1\ge t_0$ such that agent $i$ leaves agent~$j$ for $u^{t_0}_i(j) - c(u^{t_0}_i,j)$ times between time $t_0$ and time $t_1$. 
	Since agent~$i$ is deviator-resentful, this implies that $u^{t_1}_i(j) \le c(u^{t_0}_i,j)$. 
	Additionally, deviator-resentful agents can only decrease utilities for other agents. 
	Therefore, it holds for all $t\ge t_0$ and $k\in N$ that $u^t_i(k) \le u^{t_0}_i(k)$. 
	Consequently, from the definition of $c(u^{t_0}_i,j)$ it follows that agent~$i$ cannot deviate to form  a coalition together with $j$ after $t_1$, because this deviation would not be individually rational. 
	However, such a deviation takes place because we have assumed that $i$ and $j$ deviate together at some point after $t_0$ and such a deviation is performed infinitely often. 
	This establishes that $i$ and $j$ cannot leave each other after $t_0$. 
	
	Finally, consider some time $t\geq t_0$ in which a coalition $C_t$ performs a group deviation.  
	Then, from our above observation we get that no agent that is part of $C_t$ can ever leave another agent that is part of $C_t$ after this, implying that  $\pi^{t'}(j)\subseteq C_t$ for all $j\in C_t$ and $t'\geq t$. 
	This implies that the size of $C_t$ needs to monotonically increase over time.
	Notably, this holds for all coalitions that ever performed a deviation after time $t_0$. 
	Yet the overall number of agents is bounded, resulting in a contradiction. 
\end{proof}

As $\mathit{AS}_i$ and $\mathit{MF}_i$ satisfy enemy domination (\Cref{prop:allaxioms}), \Cref{thm:dev-resent:SCS-IR-converges} implies that the individually rational SCS, IS, and CS dynamics converge in ASHGs and MFHGs for deviator-resentful agents. 

\subsubsection{Absence of Individual Rationality}

If we drop the requirement that deviations are individual rational, the picture changes.
For MFHGs, IS, CS, and SCS dynamics may cycle.
\begin{proposition}\label{prop:dev-recISycle-frac}
	The IS, CS, and SCS dynamics may cycle in MFHG for deviator-resentful agents, even when starting from the singleton partition.
\end{proposition}
\newcommand{\agA}{\alpha}
\newcommand{\agB}{\beta}
\newcommand{\agC}{\gamma}
\begin{proof}

	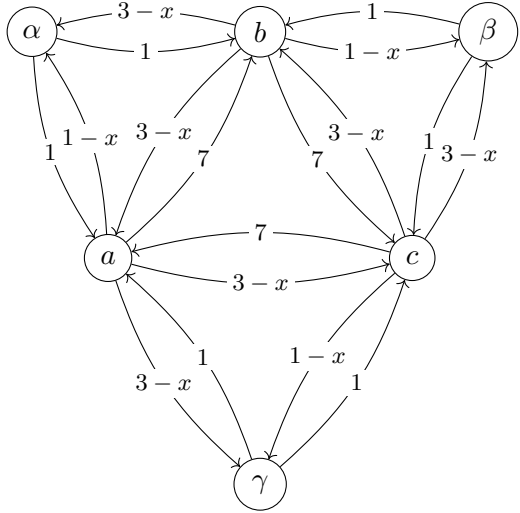
\begin{figure}
		\centering
		\begin{tikzpicture}[
		element/.style={shape=circle,draw, fill=white}
		]
		\pgfmathsetmacro\unit{4}
		\node[element] (a) at (0,0) {$a$};
		\node[element] (c) at (\unit,0) {$c$};
		\node[element] (b) at (0.5*\unit,0.75*\unit) {$b$};
		
		\node[element] (1) at (-0.25*\unit,0.75*\unit) {$\agA$};
		\node[element] (2) at (1.25*\unit,0.75*\unit) {$\agB$};
		\node[element] (3) at (0.5*\unit,-0.75*\unit) {$\agC$};
		
		\draw[->, bend right=15] (a) edge  
		node[fill=white,anchor=center, pos=0.5, inner sep =2pt] 
		{\footnotesize $7$} (b);
		\draw[->, bend right=15] (b) edge  
		node[fill=white,anchor=center, pos=0.5, inner sep =2pt] 
		{\footnotesize $3-x$} (a);
		
		\draw[->, bend right=15] (b) edge  
		node[fill=white,anchor=center, pos=0.5, inner sep =2pt] 
		{\footnotesize $7$} (c);
		\draw[->, bend right=15] (c) edge  
		node[fill=white,anchor=center, pos=0.5, inner sep =2pt] 
		{\footnotesize $3-x$} (b);
		
		\draw[->, bend right=15] (c) edge  
		node[fill=white,anchor=center, pos=0.5, inner sep =2pt] 
		{\footnotesize $7$} (a);
		\draw[->, bend right=15] (a) edge  
		node[fill=white,anchor=center, pos=0.5, inner sep =2pt] 
		{\footnotesize $3-x$} (c);
		
		\draw[->, bend right=15] (b) edge  
		node[fill=white,anchor=center, pos=0.5, inner sep =2pt] 
		{\footnotesize $3-x$} (1);
		\draw[->, bend right=15] (b) edge  
		node[fill=white,anchor=center, pos=0.5, inner sep =2pt] 
		{\footnotesize $1-x$} (2);
		
		\draw[->, bend right=15] (c) edge  
		node[fill=white,anchor=center, pos=0.5, inner sep =2pt] 
		{\footnotesize $3-x$} (2);
		\draw[->, bend right=15] (c) edge  
		node[fill=white,anchor=center, pos=0.5, inner sep =2pt] 
		{\footnotesize $1-x$} (3);
		
		\draw[->, bend right=15] (a) edge  
		node[fill=white,anchor=center, pos=0.5, inner sep =2pt] 
		{\footnotesize $3-x$} (3);
		\draw[->, bend right=15] (a) edge  
		node[fill=white,anchor=center, pos=0.5, inner sep =2pt] 
		{\footnotesize $1-x$} (1);

		\draw[->, bend right=15] (1) edge  
		node[fill=white,anchor=center, pos=0.5, inner sep =2pt] 
		{\footnotesize $1$} (b);
		\draw[->, bend right=15] (1) edge  
		node[fill=white,anchor=center, pos=0.5, inner sep =2pt] 
		{\footnotesize $1$} (a);

		\draw[->, bend right=15] (2) edge  
		node[fill=white,anchor=center, pos=0.5, inner sep =2pt] 
		{\footnotesize $1$} (b);
		\draw[->, bend right=15] (2) edge  
		node[fill=white,anchor=center, pos=0.5, inner sep =2pt] 
		{\footnotesize $1$} (c);
		
		\draw[->, bend right=15] (3) edge  
		node[fill=white,anchor=center, pos=0.5, inner sep =2pt] 
		{\footnotesize $1$} (a);
		\draw[->, bend right=15] (3) edge  
		node[fill=white,anchor=center, pos=0.5, inner sep =2pt] 
		{\footnotesize $1$} (c);

		\end{tikzpicture}
		\caption{Example from \Cref{prop:dev-recISycle-frac}. Utilities after $x$ cycles of deviations. All missing edges indicate that the utilities of the agents for each other are zero.}\label{fig:dev-recISycle-frac}
	\end{figure}
	Consider the MFHG with agent set $N = \{a,b,c,d,e,f\}$ and utilities as depicted in \Cref{fig:dev-recISycle-frac}, where the number of an arc from agents $i$ to $j$  describes the utility that $i$ has for $j$. The initial utilities result from setting $x$ to $0$ in \Cref{fig:dev-recISycle-frac}. 
	Consider the following infinite sequence of partitions. For $n\geq 0$,
	\begin{itemize}
		\item $\pi^{3n} = \{\{\agA,a,b\},\{\agB\},\{\agC,c\}\}$,
		\item $\pi^{3n+1} = \{\{\agA,a\},\{\agB,b,c\},\{\agC\}\}$, and
		\item $\pi^{3n+2} = \{\{\agA\},\{\agB,b\},\{\agC,a,c\}\}$.
	\end{itemize}
	Note that the utilities of agents after $x$ executions of this cycle of three deviations are shown in \Cref{fig:dev-recISycle-frac}. 
	Observing that for $k\ge 1$, $\pi^{k-1}$ leads to $\pi^k$ by means of a (S)CS deviation completes the counterexample.
	Note that we can also modify the dynamics to start in the singleton partition, by inserting the two partitions $\{\{\agA\},\{\agB\},\{\agC\}, \{a\},\{b\},\{c\}\}$ and $\{\{\agA,a,b\},\{\agB\},\{\agC\}, \{c\}\}$ in the beginning of the dynamics.

	Using the same initial utilities, for IS deviations we have the following cycling sequence of partitions. For $n\geq 0$, 
	\begin{itemize}
		\item $\pi^{6n} = \{\{\agA,a,b\},\{\agB\},\{\agC,c\}\}$,
		\item $\pi^{6n+1} = \{\{\agA,a,b\},\{\agB,c\},\{\agC\}\}$,
		\item $\pi^{6n+2} = \{\{\agA,a\},\{\agB,b,c\},\{\agC\}\}$,
		\item $\pi^{6n+3} = \{\{\agA\},\{\agB,b,c\},\{\agC,a\}\}$,
		\item $\pi^{6n+4} = \{\{\agA\},\{\agB,b\},\{\agC,a,c\}\}$, and
		\item $\pi^{6n+5} = \{\{\agA,b\},\{\agB\},\{\agC,a,c\}\}$.
	\end{itemize}
	Again the utilities of agents after $x$ executions of this cycle of six partitions are shown in  \Cref{fig:dev-recISycle-frac}. 
	We again observe that for $k\ge 1$, $\pi^{k-1}$ leads to $\pi^k$ by means of a IS deviation, which completes the counterexample.
	Again we can also modify the example to start in a singleton partition by adding the three partitions $\{\{\agA\},\{\agB\},\{\agC\}, \{a\},\{b\},\{c\}\}$, $\{\{\agA,a\},\{\agB\},\{\agC\}, \{b\}, \{c\}\}$, $\{\{\agA,a,b\},\{\agB\},\{\agC\}, \{c\}\}$ in the beginning of the dynamics. 
\end{proof}

By contrast, IS dynamics are still always guaranteed to converge in ASHG.
Note that this is the only contrast between ASHGs and MFHGs proven in this paper (in all other cases we either have the same result for both classes, or a result for one and an open question for the other).
\begin{theorem}
	In ASHGs, the IS dynamics converges for deviator-resentful agents. 
\end{theorem}
\begin{proof}
	Let an ASHG $(N,u^0)$ be given. 
	For the sake of contradiction assume that there exists an infinite sequence $(\pi^t)_{t\ge 0}$ of partitions resulting from IS deviations of deviator-resentful agents.
	For each $t\geq 0$, let $d^t$ be the agent deviating in step $t$ from $\pi^{t-1}(d^t)$ to $\pi^{t}(d^t)$. 
	By \Cref{lem:infoccurence}, there exists $t_0\geq 0$ such that every deviation performed after $t_0$ is performed infinitely often.
	
	Note that for each $t\geq t_0$ it holds that $d^t$ cannot be left by an agent $j\in \pi^{t}(d^t)$ after $t_0$, as otherwise $j$ will leave $d^t$ infinitely often and thus, as utilities are only decreasing, at some point will no longer approve the join of $d^t$, as she derives negative utility from $d^t$. 
	We refer to this as the first observation. 
	
	Fix some $t\geq t_0$ and let $d:=d_t$ and $C:=\pi^{t}(d^t)$ (note that $d$ joins $C\setminus \{d^t\}$ at step $t$ ). As our second observation we now show that there is some $t'> t$ with $d^{t'}=d$ and $\pi^{t'-1}(d^t) = C$. 
	
	From our choice of $t_0$ it follows that there is some $t'\geq t$ where $d$ performs a deviation for the next time. 
	Assume now, for the sake of contradiction that there is some $t''$ with $t'> t''> t$ where $\pi^{t''}(d)\neq \pi^{t''-1}(d)$ and select the smallest such $t''$. 
	If $\pi^{t''}(d)$ is changed because an agent $j$ left it, then, by our choice of $t''$, it follows that $j\in C$, which contradicts our first observation. 
	If $\pi^{t''}(d)$ is changed because an agent $j$ joined it, then it needs to hold that $j$ leaves the coalition of $d$ again before $t'$:
	Otherwise, $d$ leaves $j$ which again contradicts the first observation as $d$ needs to approve the join of $j$. Thus, the second observation follows. 
	
	However, the second observation implies that for $d$ it holds that $u_d^{t}(\pi^t(d))=u_d^{t'-1}(\pi^{t'-1}(d))$.
	In the next step, $d$ performs an IS deviation and thus increases her utility, i.e., $u_d^{t'}(\pi^{t'}(d))>u_d^{t'-1}(\pi^{t'-1}(d))$. 
	Afterwards, the second observation can be applied again until $d$ performs the next deviation. 
	Thus, $d$'s utility is strictly increasing, however, as utilities are initially bounded and resent can only cause their decay, this leads to a contradiction. Hence, IS dynamics have to converge. 
\end{proof}

We have seen that deviator-resent can be powerful force for stability, as CS dynamics with individual rational deviations in MFHGs and ASHGs and IS dynamics in ASHGs always converge. 
However, deviator-resent is not sufficient to guarantee convergence of IS and general CS dynamics in MFHGs, yielding a different behavior of ASHGs and MFHGs for IS dynamics.
Notably, we did not prove any such contrasts in our analysis of resent and appreciation.
Overall, our results indicate that deviator-resent has clear ramifications on convergence guarantees, yet the general picture seems to be slightly more nuanced than for resent or appreciation.
In particular, we were not able to settle whether SCS or CS dynamics are guaranteed to converge in ASHG for deviator-resentful agents without the individual rationality assumption, leaving this as an open question. 

\section{Simulations}\label{app:simulations}

In this section, we analyze by means of simulations how resent and appreciation influence dynamics in ASHGs.
We focus on NS dynamics, as in randomly sampled ASHGs the IS dynamics typically converge quickly even without resent or appreciation (implying that they have only a small effect).
Moreover, executing core dynamics is computationally too costly, as already checking whether an outcome is core stable is computationally intractable. 
\subsection{Setup}
We mostly focus on ASHGs with an agent set $\ag$ containing $n=50$ agents,\footnote{We analyze the influence of the number of agents in \Cref{sub:varying_n}.} and sample their utilities using one of the following two models: 
\begin{description}
	\item[Uniform] For two agents $a, b\in \ag$ with $a\neq b$, we sample $u_a(b)$, i.e., $a$'s value for $b$, by drawing a random integer between $-100$ and $100$.  
	\item[Gaussian] For each agent $a\in \ag$, we sample her \emph{base qualification} $\mu_a$ by drawing a random integer between $-100$ and $100$. 
	For two agents $a, b\in \ag$ with $a\neq b$, we sample $u_a(b)$ by drawing an integer from the Gaussian distribution with mean $\mu_b$ and standard deviation $10$.\footnote{We analyze the influence of the chosen standard deviation in \Cref{sub:varying_sig}.} 
\end{description}
Our dynamics start with the singleton partition. 
Subsequently, in each step, we compute all possible NS deviations.
If there are no NS deviations, we stop; otherwise, we sample one NS deviation uniformly at random  and execute it. 
To be able to vary the ``intensity''
of the resentful/appreciative perception, we introduce a change coefficient $c$, which we typically set to $1$:
For resentful agents, if an agent $a$ deviates from a coalition $C'$ to a coalition $C$, then we reduce the utility that agents from $C'$ have for $a$ by $c$; for appreciative agents, we increase the utility that agents from $C$ have for $a$ by $c$.
We also examine what happens if agents are both resentful and appreciative and both of the above described effects are present. 
In this case we speak of resentful-appreciative agents.
For all our simulations, we set a time-out of $\num{100000}$, i.e., after $\num{100000}$ steps we report that the dynamics did not converge.\footnote{We want to remark that this does not necessarily imply that no NS stable outcome exists in such a game or that there is no path to stability for the dynamics, but rather that selecting NS deviations randomly was not sufficient to ensure convergence (in a reasonable time).}
\begin{figure*}[t!]
	\centering
	\begin{subfigure}[t]{0.4\textwidth}
		\centering
		\resizebox{0.8\textwidth}{!}{
\begin{tikzpicture}[every plot/.append style={line width=3.5pt}]

\definecolor{color0}{rgb}{0.12156862745098,0.466666666666667,0.705882352941177}
\definecolor{color1}{rgb}{1,0.498039215686275,0.0549019607843137}
\definecolor{color2}{rgb}{0.172549019607843,0.627450980392157,0.172549019607843}

\begin{axis}[
legend cell align={left},
legend style={fill opacity=0.8, draw opacity=1, text opacity=1, at={(0.03,0.97)}, anchor=north west, draw=white!80!black,font=\LARGE},
tick align=outside,
tick pos=left,
x grid style={white!69.0196078431373!black},
xmin=-0.5, xmax=10.5,
xtick style={color=black},
y grid style={white!69.0196078431373!black},
ymin=-2995.527, ymax=62906.067,
ytick style={color=black},
ylabel={average number of steps until convergence},
xlabel={change coefficient $c$},
every tick label/.append style={font=\LARGE}, 
label style={font=\LARGE}
]
\addplot [semithick, color0]
table {%
1 59910.54
2 31032.73
3 21365.7
4 16490.6
5 13582.95
6 11644.69
7 10227.23
8 9195.22
9 8322.71
10 7706.16
};
\addlegendentry{resent}
\addplot [semithick, color1]
table {%
1 14600.44
2 7340.48
3 4796.83
4 3525.08
5 2744.03
6 2251.7
7 1881.24
8 1598.62
9 1340.96
10 1180.05
};
\addlegendentry{resent+apprec}
\addplot [semithick, color2]
table {%
1 4243.22
2 2235.57
3 1517.93
4 1150.2
5 932.43
6 817.8
7 706.35
8 634.38
9 569.47
10 522.65
};
\addlegendentry{apprec}
\end{axis}

\end{tikzpicture}}
		\caption{Uniform utilities}\label{fig:change-coeff-random}
	\end{subfigure}\qquad \qquad 
	\begin{subfigure}[t]{0.4\textwidth}
		\centering
		\resizebox{0.9\textwidth}{!}{
\begin{tikzpicture}[every plot/.append style={line width=3.5pt}]

\definecolor{color0}{rgb}{0.12156862745098,0.466666666666667,0.705882352941177}
\definecolor{color1}{rgb}{1,0.498039215686275,0.0549019607843137}
\definecolor{color2}{rgb}{0.172549019607843,0.627450980392157,0.172549019607843}

\begin{axis}[
legend cell align={left},
legend style={fill opacity=0.8, draw opacity=1, text opacity=1, at={(0.03,0.97)}, anchor=north west, draw=white!80!black,font=\LARGE},
tick align=outside,
tick pos=left,
x grid style={white!69.0196078431373!black},
xmin=-0.5, xmax=10.5,
xtick style={color=black},
y grid style={white!69.0196078431373!black},
ymin=193.8845, ymax=1378.0455,
ytick style={color=black},
ylabel={average number of steps until convergence},
xlabel={change coefficient $c$},
every tick label/.append style={font=\LARGE}, 
label style={font=\LARGE},
ylabel shift = 5pt
]
\addplot [semithick, color0]
table {%
1 1020.64
2 681.38
3 568.86
4 517.1
5 492.17
6 477.12
7 463.14
8 448.99
9 447.14
10 442.01
};
\addlegendentry{resent}
\addplot [semithick, color1]
table {%
1 763.09
2 512.03
3 411.68
4 359.45
5 331.4
6 301
7 287.25
8 275.95
9 258.44
10 247.71
};
\addlegendentry{resent+apprec}
\addplot [semithick, color2]
table {%
1 1324.22
2 745.99
3 549.95
4 454.22
5 390.16
6 352.34
7 319.94
8 294.55
9 276.99
10 259.03
};
\addlegendentry{apprec}
\end{axis}

\end{tikzpicture}}
		\caption{Gaussian utilities}\label{fig:change-coeff-gauss}
	\end{subfigure}%
	\caption{Influence of the change coefficient on the average convergence time} \label{fig:change-coeff}
\end{figure*}
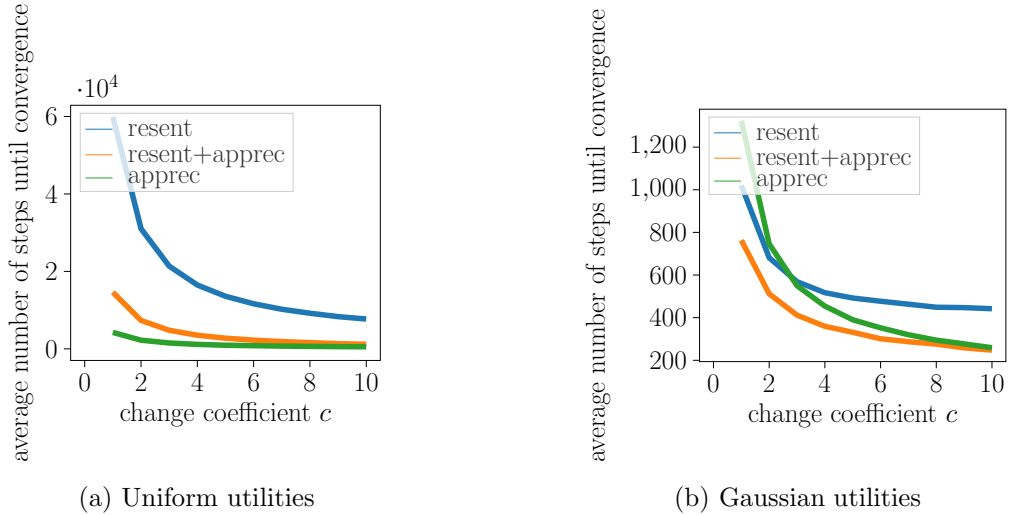
\subsection{Convergence Time}
We start by analyzing the influence of resent and of appreciation on how fast NS dynamics converge.
For this, we sampled $100$ games with $50$ agents and for each recorded the number of steps until convergence.

\paragraph{Uniform Utilities}
In \Cref{fig:change-coeff-random}, we visualize the results for uniform utilities. 
The NS dynamics for agents that are neither resentful nor appreciative did not converge in any of our sampled games (within the limit of $\num{100000}$ steps). 
By contrast, even for a change coefficient $c=1$, NS dynamics converged in all games for resentful or appreciative agents.
However, there is a clear difference between these two: 
For resentful agents, the average number of steps until convergence is $\num{59910}$ for $c=1$, whereas for appreciative agents the NS dynamics converges much faster (for $c=1$ the average convergence time is $\num{4243}$).
Increasing $c$ to $4$, for resentful agents and for appreciative agents, the average convergence time roughly quarters, while increasing $c$ to $10$ only decreases the time by an additional factor of two. 
In sum, appreciation seems more helpful to establish fast convergence than resent in case of uniform utilities.
Nevertheless, for both concepts, the number of steps until convergence is quite large (compared to the number of agents).
While increasing the change coefficient leads to faster convergence the decrease in convergence time is particularly strong for smaller values of the coefficient, indicating that NS dynamics need some time to find the ``right'' deviations somewhat independent of the value of the change coefficient. 

In contrast to our theoretical analysis, we also consider the case of resentful-appreciative agents. 
Compared to appreciative agents, adding resent leads to a substantial increase of the convergence time by a factor of around $2.5$ (while for resentful agents, adding appreciation still leads to faster convergence).
While this slower convergence for resentful-appreciative agents compared to appreciative agents may be surprising at first glance, recall the intuitive justifications why resent and appreciation contribute to a faster convergence.
For appreciative agents, utilities only increase over time, whereas for resentful agents utilities only decrease over time. 
Thus, if we combine the two, it is in principle possible that some valuations that increase for appreciative agents stay constant for resentful-appreciative agents, and the two effects can cancel out each other. 
For uniform utilities, this effect seems to be stronger than the additional ``stability force'' established by resent.

\paragraph{Gaussian Utilities}
In \Cref{fig:change-coeff-gauss}, we visualize the result of our first set of simulations for Gaussian utilities. 
In this case, the NS dynamics for agents that are neither resentful nor appreciative converged in $3$ of the $100$ games.  
In contrast, for resentful or for appreciative agents, the NS dynamics converged in all games.
In particular, convergence was much quicker (in at most $\num{2000}$ steps) than under uniform utilities, indicating that ASHGs under Gaussian utilities seems to facilitate reaching stable states compared to uniform utilities. 
Examining the results in more detail, the difference between resentful and appreciative agents is less profound here than for uniform utilities.
While resent leads to faster convergence for $c<3$, appreciation is more powerful for $c\geq 3$. 
Considering the influence of the change coefficient, for resentful agents, we again see the trend that increasing the change coefficient has a strong effect for smaller $c$ but that this effect becomes less strong for larger $c$. 
On the other hand, for appreciative agents, the relation between $c$ and the convergence time is rather linear. 
Moreover, in contrast to uniform utilities, resent and appreciation seem to  not ``cancel out'' each other.
For resentful-appreciative agents, NS dynamics converge faster than for either of the two separately.

\begin{table*}[t!]\centering
	\caption{Properties of different dynamics and the produced outcomes for different ASHGs. For an explanation of the table, we refer to the beginning of  \Cref{subsec:structural}.}\label{table:experiment}
	
	\begin{adjustbox}{max width=\textwidth}
		\begin{tabular}{@{}llllllllllllll@{}}
			\toprule
			& 
			\multicolumn{1}{c}{} & 
			\phantom{a}&
			\multicolumn{3}{c}{coalition sizes} & 
			\phantom{a}& 
			\multicolumn{2}{c}{stability in original} & 
			\phantom{a} & 
			\multicolumn{3}{c}{utilities} 
			\\ 
			\cmidrule{4-6} 
			\cmidrule{8-9}
			\cmidrule{11-13}
			& steps && number & avg. size & max size && \# ag. IR viol & \# ag. NS dev && avg. util & avg. change & pos. vals 
			\\
			\midrule
			\makebox[-1pt][l]{\textbf{Uniform utilities $\mathbf{n=50}$} }\\
			resent & 60055.0 && 50 & 1.0 & 1 && 0 & 50 && 0.0 & -51.25 & 0.0  \\
			apprec & 4309.0 && 2.74 & 19.31 & 42.57 && 18.7 & 22.21 && 589.39 & 12.8 & 0.56 \\
			resent+apprec & 15261.0 && 5.26 & 9.73 & 18.43 && 2.33 & 10.15 && 237.08 & -0.13 & 0.51 
			\\
			\midrule
			\makebox[-1pt][l]{\textbf{Gaussian utilities $\mathbf{n=50}$} }\\
			resent & 968.0 && 25.74 & 1.99 & 25.2 && 0 & 5.07 && 626.3 & -0.42 & 0.49  \\
			apprec & 1226.0 && 21.69 & 2.36 & 25.19 && 3.67 & 5.11 && 638.79 & 0.6 & 0.51 \\
			resent+apprec & 694.0 && 24.64 & 2.07 & 25.28 && 0.61 & 4.66 && 637.34 & -0.12 & 0.5 \\
			\midrule
			\midrule
			\makebox[-1pt][l]{\textbf{Gaussian utilities $\mathbf{n=25}$} }
			\\
			resent & 282.0 && 13.62 & 1.91 & 12.36 && 0 & 2.11 && 291.49 & -0.49 & 0.48  \\
			apprec & 438.0 && 11.92 & 2.2 & 12.29 && 1.76 & 2.24 && 296.59 & 0.77 & 0.5 \\
			resent+apprec & 208.0 && 13.32 & 1.95 & 12.42 && 0.18 & 1.97 && 297.54 & -0.12 & 0.49 \\
			\midrule
			\makebox[-1pt][l]{\textbf{Gaussian utilities $\mathbf{n=25}$ with stable outcome (13 games)} }
			\\
			original & 100.38 && 12.46 & 2.15 & 13.54 && 0 & 0 && 344.26 & 0.0 & 0.52 \\
			resent & 90.69 && 12.46 & 2.15 & 13.54 && 0 & 0 && 343.14 & -0.16 & 0.52 \\
			apprec & 92.77 && 12.31 & 2.24 & 13.54 && 0.15 & 0.15 && 348.78 & 0.29 & 0.52 \\
			resent+apprec & 100 && 12.46 & 2.15 & 13.54 && 0 & 0 && 347.82 & -0.05 & 0.52\\
			\bottomrule
		\end{tabular}
	\end{adjustbox}
\end{table*}

\subsection{Structure of Outcomes and Comparison to Base Game} \label{subsec:structural}
We now take a closer look at the outcomes and utility functions produced by NS dynamics for resentful and/or appreciative agents. 
Again, we generated $100$ games with $n=50$ agents each for uniform and Gaussian utilities.
In addition, we generated $100$ games with $n=25$ agents and Gaussian utilities (as the original NS dynamics converge more often in such games).
In all games, we set the change coefficient $c$ to $1$.
\Cref{table:experiment} summarizes the results of our simulations.
All values in the table are averaged from the respective $100$ games.
In the last part of the table, we only consider the $13$ games with $n=25$ and Gaussian utilities in which an NS dynamics in the original game converged within \num{100000} steps. 
For reference, we depict the number of steps the dynamics needed to converge in the first column.
We analyze the structure of the produced NS outcomes as follows. Columns three to six consider the produced coalitions, that is, the number of coalitions and their average and maximum size. 
The next two columns concern the outcome's degree of stability with respect to the original utilities, where we record the number of agents violating individual rationality and possessing an NS deviation, respectively.
The last three columns concern the change of the utility profile, that is, the average utility in the outcome partition with respect to the final utilities, the average change of each entry of the utility function comparing the initial and final utilities, and the fraction of pairs of friends with respect to the final utilities, that is tuples $(a,b)\in \ag^2$ for which $u_a(b) > 0$.\footnote{Note that, in both utility models, this value is on average $0.5$ for the initial utilities.}

\paragraph{Uniform Utilities}
We first focus on uniform utilities. Recall that our dynamics need many steps until reaching convergence. 
Therefore, it is not surprising that the produced outcomes are quite ``degenerated'' for both resentful agents and appreciative agents. 
For resentful agents, the produced outcomes consist only of singleton coalitions in all games, which means that agents have left each other sufficiently often to ensure that all pairwise utilities are non-positive.
This is also reflected by the facts that all agents have NS deviations with respect to their original utilities, and that on average the valuations of agents changed by $51.25$. 
Overall, the produced outcomes have little connection to the original game and simply exploit that all utilities become negative at some point for resentful agents.

By contrast, for appreciative agents, there is typically one large coalition containing $40$ or more agents together with one or two small coalitions. 
Indeed, this is also reflected by the observations that the average number of coalitions is $2.74$ and the average maximum size is $42.57$. 
The typical run of an NS dynamics for appreciative agents here can be described in two phases. First, agents increase their utility for each other by deviating between smaller coalitions (where some agents, which are negatively valued by many others, are not joined, which often leads to them being part of small coalitions in the final outcome).  
Subsequently, in a second phase, agents already possess a generally high utility level, and thus tend to favor large coalitions (even when having a negative utility for some of the agents in the coalition).
Notably, it does not happen that eventually all utilities between pairs of agents are positive, and in this sense, the behavior of appreciative agents is not the contrary of the behavior of resentful agents. 
In fact, slightly counterintuitively, only $56\%$ of agent pairs have a positive evaluation after convergence of the dynamics.
Consequently, agents (from the large coalition) often dislike other coalition members: 
On average, an agent only values $59\%$ of her coalition members positively. 
Nevertheless, the relationship of the produced outcome and the agent's original utilities is still quite low with $18.7$ agents for which individual rationality is violated and $22.21$ agents having an NS deviation.

For resentful-appreciative agents, the produced outcomes are in some sense between the two extremes for resentful agents and for appreciative agents: 
Typically, several medium-size coalitions form (the average number of coalitions is $5.26$ and the maximum size is on average $18.43$). 
Further, the average change of the utility values is $-0.13$ and only $51\%$ of agent pairs have a positive evaluation, implying that the utility changes caused by resent and by appreciation cancel out each other from an aggregated perspective.
However, on an individual level, utilities still change quite drastically, as the absolute difference between the initial utilities values and the values at the end of the dynamic is on average $12.15$.
Notably, the outcome produced by resentful-appreciative agents is also in some sense less ``degenerated'' as for the the two separately: 
It consists of medium-size coalitions, utilities are structurally more similar to the initial utilities, and most importantly, the outcome is closer to stability in the original game (with only $2.33$ agents for which individual rationality is violated and only $10.15$ agents having an NS deviation).

\paragraph{Gaussian Utilities}

We now turn to Gaussian utilities, where the dynamics converge much quicker than for uniform utilities. 
Moreover, the outcomes produced by our three dynamics are quite similar, which follows the intuition that the final utility profiles remain quite similar after the execution of ``few'' steps.
In general, NS outcomes produced by our dynamics typically consist of one large coalition containing roughly half of the agents (these are usually the agents with positive ground qualification), while other agents are placed into coalitions of size one or two. 

Let us first focus on the second and third part of \Cref{table:experiment} presenting the results of $100$ games with $n=50$ and $n=25$ agents, respectively.
First, the outcomes for appreciative agents typically contain fewer coalitions than for resentful agents. 
Interestingly, this is not achieved because of the size of the ``large'' coalition but due to a larger average size of the many small coalitions. 
For resentful-appreciative agents, the produced outcomes are structurally similar to the ones for resentful agents with a slightly smaller number of coalitions (achieved by an increase of the size of the ``large'' coalition).
Second, for all three types of dynamics the produced outcomes are much closer to being stable in the original game than for uniform utilities (only around $10\%$ of the agents have an NS deviation). 
In particular, for resentful-appreciative agents, the produced outcomes are closest to stability in the original game. 
Lastly, considering the average utility of agents in the produced outcome, the three perception types produce quite similar results: 
Naturally, for appreciative agents, the average utility of agents in the produced outcome is highest.
However, in both other dynamics, agents are only slightly less happy.

In the fourth part of the table, we depict statistics concerning the $13$ of our $100$ games with $n=25$ agents having Gaussian utilities for which an original NS dynamic converged (in the time limit). 
In the first line, we show properties of the outcomes produced by the original dynamics. 
The NS outcomes produced by the original NS dynamics are structurally quite similar to the ones shown in the third part of the table, indicating that the outcomes produced by our three dynamics on games with $25$ or $50$ agents having Gaussian utilities are quite ``natural'' and not degenerated. 
The similarities become even more profound when comparing the outcomes produced by the original dynamics on the selected $13$ games to the outcomes produced by our three types of dynamics on the same $13$ games: 
For resentful and resentful-appreciative agents,
the produced outcomes are very similar or even identical to the outcome produced by the original NS dynamics in all $13$ games (and are in particular always stable in the original games). 
For appreciative agents, the ``large'' coalition in the outcome is typically quite similar or identical to the ``large'' coalition produced by the original dynamics; however, sometimes fewer small coalitions are produced (in particular, some of the produced outcomes are not stable in the original game). 
Concerning the average utility in the final partition, resentful agents are only marginally less happy than in the original dynamics, indicating that only very few agents end up in a coalition with an agent they left at some point, whereas the average utility is slightly higher for appreciative agents and resentful-appreciative agents.

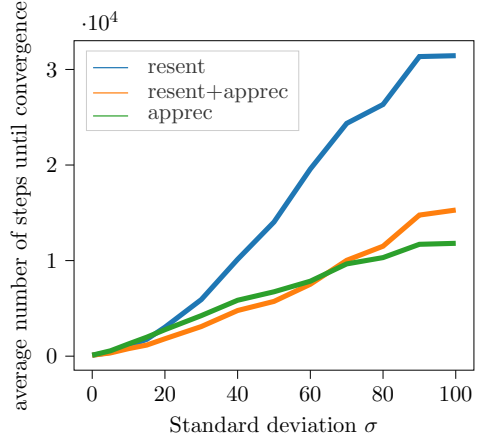
\begin{figure}[t]
	\centering
	\resizebox{0.4\textwidth}{!}{
\begin{tikzpicture}[every plot/.append style={line width=2.5pt}]

\definecolor{color0}{rgb}{0.12156862745098,0.466666666666667,0.705882352941177}
\definecolor{color1}{rgb}{1,0.498039215686275,0.0549019607843137}
\definecolor{color2}{rgb}{0.172549019607843,0.627450980392157,0.172549019607843}

\begin{axis}[
legend cell align={left},
legend style={fill opacity=0.8, draw opacity=1, text opacity=1, at={(0.03,0.97)}, anchor=north west, draw=white!80!black},
tick align=outside,
tick pos=left,
x grid style={white!69.0196078431373!black},
xmin=-5, xmax=105,
xtick style={color=black},
y grid style={white!69.0196078431373!black},
ymin=-1466.9635, ymax=33008.2135,
ytick style={color=black},
ylabel={average number of steps until convergence},
xlabel={Standard deviation $\sigma$}
]
\addplot [semithick, color0]
table {%
0 102.22
5 389.8
10 1079.06
15 1757.62
20 3031.26
30 5920.41
40 10132.1313131313
50 14034.96
60 19584.1612903226
70 24355.7922077922
80 26335.0576923077
90 31348.1081081081
100 31441.16
};
\addlegendentry{resent}
\addplot [semithick, color1]
table {%
0 100.09
5 337.73
10 788.25
15 1161.13
20 1821.33
30 3098.81
40 4776.37
50 5730.03
60 7523.81
70 10029.04
80 11505.55
90 14762.68
100 15290.66
};
\addlegendentry{resent+apprec}
\addplot [semithick, color2]
table {%
0 101.61
5 559.99
10 1282.87
15 1985.18
20 2770.62
30 4241.53
40 5850.48
50 6736.6
60 7846.21
70 9652.96
80 10314.18
90 11704.36
100 11807.44
};
\addlegendentry{apprec}
\end{axis}

\end{tikzpicture}}
	\caption{Influence of the standard deviation of Gaussian utilities on the average convergence time. We only take into account dynamics converging within the time limit of $\num{100000}$ steps. For resentful agents, the NS dynamics did not converge in $7/23/48/63/75$ games for a standard deviation of $60/70/80/90/100$, respectively (for appreciative agents and for resentful-appreciative agents the dynamics always converged).} \label{fig:varying-sigma}
\end{figure}

\subsection{Influence of Standard Deviation for Gaussian Utilities}\label{sub:varying_sig}
In this section, we focus on Gaussian utilities and analyze the influence of the standard deviation $\sigma$ on convergence times.
Recalling the drastic differences in the behavior of the dynamics between Gaussian utilities (with $\sigma=10$) and uniform utilities observed in the previous subsections, it is to be expected that the convergence time increases with increasing $\sigma$ (as we are in some sense getting closer to uniform utilities for larger $\sigma$). 
To analyze this effect in more detail, for $\sigma\in \{0,5,10,\dots, 95,100\}$, we sampled $100$ games with $n=50$ agents having Gaussian utilities with standard deviation $\sigma$, and measured the average convergence times of our three types of NS dynamics for a change coefficient $c=1$. 
The results are shown in \Cref{fig:varying-sigma}, where we only show the average convergence time of all dynamics that converged in the time bound of $\num{100000}$ steps. 
For resentful agents, the NS dynamics did not converge in $7/23/48/63/75$ games for a standard deviation of $60/70/80/90/100$, respectively (for appreciative agents and resentful-appreciative agents the dynamics always converged).
Overall, our hypothesis that increasing the standard deviation leads to an increased convergence time is confirmed clearly. 
Slightly unexpected, for $\sigma\geq 80$, Gaussian utilities seem to be even ``harder'' than uniform utilities:
For resentful agents, the dynamics did not even converge for roughly half of the games and for appreciative agents the average convergence time is about twice as large.
The reason for this is that for large $\sigma$ utilities often fall out of the range $[-100,100]$ (from which we sample the agent's base qualification for Gaussian utilities and the utilities for uniform utilities). 
In games which are ``far away'' from stability and for which resentful agents or appreciative agents produce degenerated outcomes (as described in the previous section) an increased range of the utility values clearly contributes to a higher convergence time: 
For instance, for resentful agents, we have observed that the dynamics typically runs until all utility values are negative, which takes more time if some utility values are initially larger.

\subsection{Influence of the Number of Agents} \label{sub:varying_n}
We now briefly analyze the influence of the number of agents on the convergence time of NS dynamics with resentful and/or appreciative agents. 
For this, for different numbers of agents, we generated $100$ games and executed NS dynamics with resentful and/or appreciative agents and a change coefficient of $c=1$. 
We show the results for uniform utilities in \Cref{fig:varying_n-Uniform} and the results for Gaussian utilities in \Cref{fig:varying_n-Gauss}.
Overall, unsurprisingly, the higher the number of agents, the higher is the convergence time of our dynamics.

\begin{figure*}
	\centering
	\begin{subfigure}[t]{0.4\textwidth}
		\centering
		\resizebox{0.8\textwidth}{!}{
\begin{tikzpicture}[every plot/.append style={line width=2.5pt}]

\definecolor{color0}{rgb}{0.12156862745098,0.466666666666667,0.705882352941177}
\definecolor{color1}{rgb}{1,0.498039215686275,0.0549019607843137}
\definecolor{color2}{rgb}{0.172549019607843,0.627450980392157,0.172549019607843}

\begin{axis}[
legend cell align={left},
legend style={fill opacity=0.8, draw opacity=1, text opacity=1, at={(0.03,0.97)}, anchor=north west, draw=white!80!black},
tick align=outside,
tick pos=left,
x grid style={white!69.0196078431373!black},
xmin=2.25, xmax=62.75,
xtick style={color=black},
y grid style={white!69.0196078431373!black},
ymin=-4104.868, ymax=89439.968,
ytick style={color=black},
ylabel={average number of steps until convergence},
xlabel={number of agents}
]
\addplot [semithick, color0]
table {%
5 329.16
10 1803.28
15 5643.07
20 10330.8
25 15927.31
30 22577.58
35 30437.57
40 38997.91
45 48908.7
50 59877.18
55 72086.52
60 85187.93
};
\addlegendentry{resent}
\addplot [semithick, color1]
table {%
5 147.17
10 552.36
15 1352.94
20 2182.54
25 3667.8
30 5184.51
35 7433.45
40 9289.88
45 12040.03
50 14669.55
55 17724.59
60 21100.42
};
\addlegendentry{resent+apprec}
\addplot [semithick, color2]
table {%
5 168.04
10 475.57
15 931.81
20 1266.72
25 1726.53
30 2246.02
35 2772.46
40 3166.45
45 3786.61
50 4214.17
55 4823.14
60 5279.23
};
\addlegendentry{apprec}
\end{axis}

\end{tikzpicture}}
		\caption{Uniform utilities}\label{fig:varying_n-Uniform}
	\end{subfigure}~
	\begin{subfigure}[t]{0.4\textwidth}
		\centering
		\resizebox{0.87\textwidth}{!}{
\begin{tikzpicture}[every plot/.append style={line width=2.5pt}]

\definecolor{color0}{rgb}{0.12156862745098,0.466666666666667,0.705882352941177}
\definecolor{color1}{rgb}{1,0.498039215686275,0.0549019607843137}
\definecolor{color2}{rgb}{0.172549019607843,0.627450980392157,0.172549019607843}

\begin{axis}[
legend cell align={left},
legend style={fill opacity=0.8, draw opacity=1, text opacity=1, at={(0.03,0.97)}, anchor=north west, draw=white!80!black},
tick align=outside,
tick pos=left,
x grid style={white!69.0196078431373!black},
xmin=0.25, xmax=104.75,
xtick style={color=black},
y grid style={white!69.0196078431373!black},
ymin=-131.17, ymax=4287.09,
ytick style={color=black},
ylabel={average number of steps until convergence},
xlabel={number of agents}
]
\addplot [semithick, color0]
table {%
5 290.4
10 425.42
15 204.09
20 600.33
25 318
30 680.12
35 498.57
40 682.35
45 892.79
50 1050.91
55 1163.27
60 1377.46
65 1800.6
70 1925.81
75 2254.01
80 2410.9
85 2947.57
90 3625.39
95 3640.52
100 4086.26
};
\addlegendentry{resent}
\addplot [semithick, color1]
table {%
5 69.66
10 76.28
15 98.46
20 152.33
25 227.27
30 296.01
35 377.13
40 498.04
45 646.1
50 743.11
55 867.7
60 950.19
65 1189.93
70 1351.46
75 1552.42
80 1619.37
85 2008.51
90 2386.07
95 2509.22
100 2778.19
};
\addlegendentry{resent+apprec}
\addplot [semithick, color2]
table {%
5 98.86
10 140.78
15 208.04
20 306.16
25 445.74
30 573.58
35 735.43
40 974.34
45 1164.23
50 1314.99
55 1461.96
60 1588.87
65 1942.83
70 2141.31
75 2503.69
80 2623.43
85 3085.82
90 3541.88
95 3660.13
100 4043.14
};
\addlegendentry{apprec}
\end{axis}

\end{tikzpicture}}
		\caption{Gaussian utilities}\label{fig:varying_n-Gauss}
	\end{subfigure}%
	\caption{Influence of the number of agents on the average convergence time.}
\end{figure*}
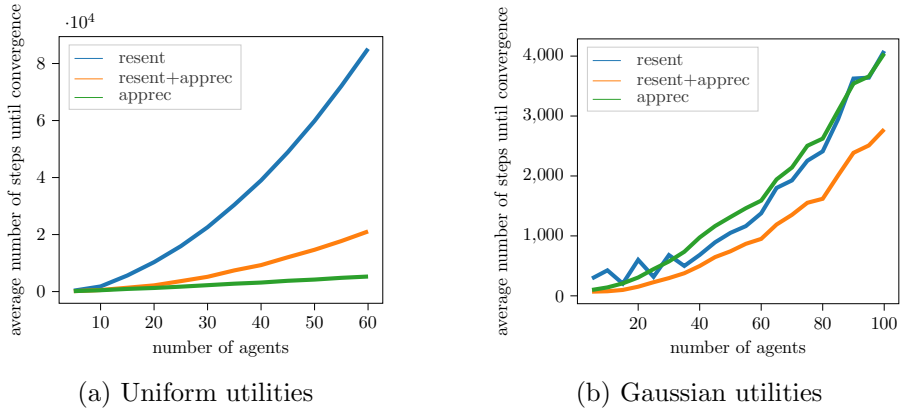

For uniform utilities, the relation of the different types of dynamics is independent of the number of agents.
For appreciative agents, the average convergence time is the lowest, while 
convergence is still faster for resentful agents than for resentful-appreciative agents.
The increase in the convergence time between $n=10$ and $n=60$ is almost linear for all three types.
However, the slope substantially depends on the considered type: 
For resentful agents, comparing $n=10$ and $n=60$, the convergence time increases by a factor of $47$, for appreciative agents the factor is $11$, and for resentful-appreciative agents the factor is $38$.
Thereby, the difference between the three becomes larger with a growing number of agents. 
This again highlights the different nature of the three types of NS dynamics for uniform utilities that we have observed above: 
For resentful agents, for uniform utilities, all utility values need to become negative, which requires substantially, but linearly, more deviations the higher the number of agents.
In contrast, for appreciative agents, agents typically only join each other until almost everyone has a positive valuation for a very large coalition.
Naturally also this requires more deviations the higher the number of agents, yet with slower scaling.

For Gaussian utilities, we also see an almost linear increase of the average convergence time.
Moreover, for most considered agent numbers, dynamics for resentful-appreciative agents converge faster than the ones for resentful agents which in turn converge faster than the ones for appreciative agents. 
However, in contrast to uniform utilities, here the difference between the three types of NS dynamics changes only slightly for different numbers of agents.
Furthermore, comparing $n=10$ and $n=60$, the convergence time of all three types of NS dynamics only increases by a factor of around $12$. 
Lastly, for resentful agents, we can observe a slightly unexpected phenomenon, as for a  smaller number of agents the convergence time does not steadily increase, e.g., for $n=10$ the average converge time is $425$ steps and for $n=15$ it is $204$ steps.

\section{Shortest Converge Sequences} \label{app:shortest}

In our simulations in \Cref{app:simulations}, we have analyzed how fast random executions of NS dynamics converge for resentful and/or appreciative agents. 
An interesting related direction to shed further light on the power of resent and appreciation is to analyze the length of the \emph{shortest} converging execution of dynamics.
The corresponding computational problem is to decide whether in a given game (where a stable outcome is guaranteed to be reachable), there is a converging deviation sequence of a given length from some given starting partition.

Notably, this problem has not been addressed in the literature for classical dynamics, yet is of no less relevance in the general case. 
While existing hardness results for deciding whether a game admits a stable outcome suggest the hardness of the shortest converge sequence problem for the general case, they do not directly imply hardness, as stable states might not be reachable from some initial partition via the allowed deviations.
Moreover, convergence might require exponential time from some initial partition \citep{BBT22a}.

We present three reductions showing that for ASHGs deciding whether we can converge in a given number of steps is NP-hard 
for CS, IS, CNS, and NS dynamics for resentful, for appreciative, and for classical agents (that is, agents with constant utility functions over time). 
More formally, for $\alpha\in \{NS,IS,CNS,CS\}$, we define the following decision problem:
\smallskip
\begin{center}
	{\small 
		\begin{tabularx}{1.0\columnwidth}{ll}
			\toprule
			\multicolumn{2}{c}{\sc{resentful-ashg-$\alpha$-sequence}} \\
			\midrule
			{\bf Given:}& \parbox[t]{0.7\columnwidth}{
				An additively separable hedonic game 
				with resentful agents and some integer~$k$.
				\vspace*{1mm}} \\
			{\bf Question:}& \parbox[t]{0.7\columnwidth}{
				Is there an execution $(\pi^t)_{t\ge 0}$ of the $\alpha$ dynamics
				that
				starts from the singleton partition $\pi^0$
				and
				converges in at most $k$ steps?
				\vspace*{.5mm}} \\ 
			\bottomrule
		\end{tabularx}
	}
\end{center}
\smallskip
Defining {\sc appreciative-ashg-$\alpha$-sequence}
analogously
for appreciative agents,
we can show that these problems are
$\NP$-hard for NS, IS, CNS, and CS.

\begin{restatable}{theorem}{sequence}
	\label{thm:sequence}
	For $\alpha\in \{NS,IS,CNS,CS\}$, {\sc resentful-} and {\sc appreciative-ashg-$\alpha$-sequence} are $\NP$-hard. 
\end{restatable}

Before turning to the proof of \Cref{thm:sequence}, we want to make a few remarks.
First, note that all our constructions in the following proofs
also work for normal agents with constant utilities over time.
Hence, our constructions also imply the hardness
of the respective decision problem for (time-independent) ASHGs.

Second, we want to discuss membership in $\NP$. 
Note that the sequence of partitions leading to a stable partition is a certificate for a Yes-instance of size $\mathcal O(nk)$.
Hence, if $k$ is given in unary encoding, then this is a certificate of polynomial size. 
Moreover, for the reduced instance in our proofs, it holds that $k\in \mathcal O(n)$, and therefore the hardness holds on a class of instances for which membership in $\NP$ is satisfied. Still, it is not clear whether {\sc resentful-} and {\sc appreciative-ashg-$\alpha$-sequence} are in $\NP$ for a general $k$ in binary encoding.

The proof of \Cref{thm:sequence} consists of three individual constructions.
We showcase the technique by proving the statement for NS and IS in \Cref{thm:sequence-NS-IS}.
The proofs for CNS and CS are similar, and are deferred to \Cref{thm:sequence-CNS,thm:sequence-CS} in the appendix.

\begin{lemma}\label{thm:sequence-NS-IS}
	For $\alpha\in \{NS,IS\}$, {\sc resentful-} and {\sc appreciative-ashg-$\alpha$-Sequence} are $\NP$-hard. 
\end{lemma}
\begin{proof}
	We first show the statement for IS and resentful agents. 
	We reduce from \textsc{Restricted Exact Cover by 3-Sets (RX3C)} where we are given a finite universe $U=\{x_1,\dots, x_{3t}\}$ and a family $\mathcal{S}=\{S_1,\dots , S_{3t}\}$ of $3$-subsets of $U$ where each element from $U$ appears in exactly three sets from~$\mathcal{S}$. 
	The question is whether there is a family $\mathcal{S}'\subseteq \mathcal{S}$ which is an exact cover of $U$. 
	
	\textbf{Construction.}
	Given 
	an instance $(U,\mathcal S)$ of \textsc{RX3C}, we 
	set $k=10t$ and
	create the following additively separable hedonic game. 
	We add one \emph{element agent} $x$ for each element $x\in U$ and one \emph{set agent} $S$ for each set $S\in \mathcal{S}$. Moreover, we add $2t$ \emph{filling agents} $F=\{f_1,\ldots,f_{2t}\}$. Lastly, we add a penalizing gadget
	consisting of $10t+1$ \emph{penalizing agents}
	$\{p\}\cup Q$ with $Q=\{p_{i,j}\mid i\in[5t],j\in [2]\}$. 
	The resulting set of agents is $N=U\cup \mathcal{S} 
	\cup F \cup \{p\}\cup Q$.
	The initial utilities of the players
	are illustrated in Figure~\ref{fig:sequence-NS-IS}:
	\begin{itemize}
		\item Each element agent has utility zero for each other element agent and
		utility~$1$ for each set agent.
		\item Each filling agent has utility~$1$ for each set agent.
		\item Each set agent $S\in \mathcal{S}$ has 
		utility $20t$ for the three element agents $x\in S$,  
		utility $60t$ for each filling agent, 
		utility~$60t$ for~$p$, and
		utility zero for each penalizing agent in $Q$.
		\item $p$ has 
		utility $10t$ for each set agent and
		utility zero for each penalizing agent in $Q$.
		\item Each penalizing agent $p_{i,j}\in Q$
		has utility $30t$ for each set agent,
		utility $30t$ for $p$, 
		utility $40t$ for her corresponding penalizing agent $p_{i,k}$ with $k\in [2], k\neq j$, and
		utility zero for all remaining penalizing agents in $Q$.
		\item All not explicitly mentioned utilities are $-1000t$.
	\end{itemize}
	\begin{figure}
		\centering
		\begin{tikzpicture}[
		element/.style={shape=circle,draw, fill=white}
		]
		\pgfmathsetmacro\xdist{2}
		\pgfmathsetmacro\ydist{-1}
		
		\node (x1) at (0,0.4*\ydist) {${x_1}$};
		\node (x3t) at (0,1.6*\ydist) {${x_{3t}}$};
		\node at (barycentric cs:x1=1,x3t=.7) {$\vdots$};
		\node (xphan) at (0,2.1*\ydist) {\phantom{${x_{3t}}$}};

		\node (S1) at (\xdist,0.4*\ydist) {$S_1$};
		\node (S3t) at (\xdist,1.6*\ydist) {$S_{3t}$};
		\node at (barycentric cs:S1=1,S3t=.7) {$\vdots$};
		\node (Sphan) at (\xdist,1.8*\ydist) {\phantom{$S_{3t}$}};
		\node (Sphan2) at (0.9*\xdist,1.6*\ydist) {\phantom{$S_{3t}$}};
		
		\node (f1) at (2*\xdist,0.4*\ydist) {$f_{1}$};
		\node (f2t) at (2*\xdist,1.6*\ydist) {$f_{2t}$};
		\node at (barycentric cs:f1=1,f2t=.7) {$\vdots$};
		\node (fphan) at (2*\xdist,1.8*\ydist) {\phantom{$f_{2t}$}};

		\node (p) at (1*\xdist,3.5*\ydist) {$p$};
		
		\node (p11) at (-2,3*\ydist) {$p_{1,1}$};
		\node (p12) at (0,3*\ydist) {$p_{1,2}$};

		\node (p5t1) at (-2,4*\ydist) {$p_{5t,1}$};
		\node (p5t2) at (0,4*\ydist) {$p_{5t,2}$};

		\node (dotsp1) at (barycentric cs:p11=1,p5t1=.7) {$\vdots$};
		\node (dotsp2) at (barycentric cs:p12=1,p5t2=.7) {$\vdots$};
		
		\node (E) [draw, dashed, rounded corners=3pt, fit={(x1) (x3t)}] {};
		\node at (x1) [above=0.35cm] {\footnotesize$0$};
		
		\node (E) [draw, dashed, rounded corners=3pt, fit={(p11) (p5t2)}] {};
		\node at (p11) [above=0.35cm] {\footnotesize$0$};
		
		\draw[<->] (p11) edge
		node[fill=white,anchor=center, pos=0.5, inner sep =2pt, above] {\footnotesize $40t$} (p12);
		\draw[<->] (p5t1) edge
		node[fill=white,anchor=center, pos=0.5, inner sep =2pt, above] {\footnotesize $40t$} (p5t2);
		\draw[<->] (p5t1) edge (p5t2);
		
		\draw[->, bend right=15] (p12) edge
		node[fill=white,anchor=center, pos=0.5, inner sep =2pt] {\footnotesize $30t$} (Sphan2);
		\draw[->, bend right=15] (Sphan2) edge
		node[fill=white,anchor=center, pos=0.5, inner sep =2pt] {\footnotesize $0$} (p12);

		\draw[->, bend right=15] (dotsp2) edge
		node[fill=white,anchor=center, pos=0.5, inner sep =2pt] {\footnotesize $30t$} (p);
		\draw[->, bend right=15] (p) edge
		node[fill=white,anchor=center, pos=0.5, inner sep =2pt] {\footnotesize $0$} (dotsp2);
		
		\draw[->, bend right=30] (Sphan) edge
		node[fill=white,anchor=center, pos=0.5, inner sep =2pt] {\footnotesize $60t$} (p);
		\draw[<-,bend left=30] (Sphan) edge
		node[fill=white,anchor=center, pos=0.5, inner sep =2pt] {\footnotesize $10t$} (p);

		\draw[->] (0.2,0.65*\ydist) -- (0.8*\xdist,0.65*\ydist) node[fill=white,anchor=center, pos=0.5, inner sep =2pt] {\footnotesize $1$};
		
		\draw[<-] (0.2,1.35*\ydist) -- (0.8*\xdist,1.35*\ydist) node[fill=white,anchor=center, pos=0.5, inner sep =2pt, above] {\footnotesize $20t$}
		node[fill=white,anchor=center, pos=0.5, inner sep =2pt, below] {\footnotesize $x\in S$};
		
		\draw[<-] (1.2*\xdist,0.8*\ydist) -- (1.8*\xdist,0.8*\ydist) node[fill=white,anchor=center, pos=0.5, inner sep =2pt] {\footnotesize $1$};
		\draw[->] (1.2*\xdist,1.2*\ydist) -- (1.8*\xdist,1.2*\ydist) node[fill=white,anchor=center, pos=0.5, inner sep =2pt] {\footnotesize $60t$};
		
		\end{tikzpicture}
		\caption{Scheme of the additively separable hedonic games in  \Cref{thm:sequence-NS-IS}. 
			All omitted utilities are $-1000t$.\label{fig:sequence-NS-IS}}
	\end{figure}
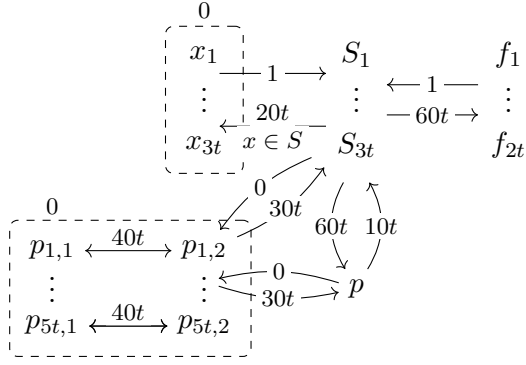 
	
	We now show that there is an exact cover of $U$
	if and only if there is a sequence starting with the singleton partition and reaching an IS partition after at most $k$ IS deviations, where agents are resentful.
	
	{\bfseries ($\Rightarrow$)}
	Given an exact cover $\mathcal{S}'\subseteq \mathcal{S}$, we let the agents deviate as follows. 
	For each $S\in \mathcal{S}\setminus \mathcal{S}'$, one filling agent deviates to $S$ (there are $2t$ of these sets and  $2t$ filling agents).
	For each $S\in \mathcal{S}'$, the three element agents $x$ with $x\in S$ deviate to~$S$.
	As $\mathcal{S'}$ is an exact cover, each filling and each element agent deviate exactly once, leading to $5t$ deviations so far. 
	Furthermore, for $i\in[5t]$,
	$p_{i,2}$ deviates to $p_{i,1}$.
	Thus, we have $10t$ deviations in total.
	If we order these deviations arbitrarily, then all deviations are NS deviations, as each agent deviates from the singleton coalition to a coalition for which she has positive utility. 
	Furthermore, these NS deviations are IS deviations as the agents 
	in the joined coalitions
	have non-negative utilities for the deviators.
	It remains to show that the resulting partition is NS (and thus also IS): 
	Observe that there is at most one set agent per coalition.
	As element and filling agents only derive a positive utility from set agents and are in a coalition with a set agent, no element or filling agent has an NS deviation. 
	Agent~$p$	only has a positive utility for the set agents but
	all set agents are accompanied by either element or filling agents
	for which $p$ has a large negative utility.
	Hence, $p$ has no NS deviation.
	Agents $p_{i,1}$ and $p_{i,2}$ with $i\in[5t]$
	do not want to deviate as they
	have utility $40t$ for coalition $\{p_{i,1},p_{i,2}\}$ and value no other available coalition with more than $30t$.
	Lastly, each set agent $S\in \mathcal{S}$ is in a coalition from which she derives utility $60t$. 
	All other coalitions are either the coalition just containing
	$p$ (for which $S$ has utility $60t$ as well), coalitions containing other penalizing agents (for which $S$ has utility zero), or contain a different set agent (for which $S$ has large negative utility). 
	Thus, $S$ has no NS deviation either and the created partition is NS.

	{\bfseries ($\Leftarrow$)}
	Assume that we have a sequence of $\ell\leq k=10t$ IS deviations leading from the singleton partition $\pi^0$ to an IS partition $\pi^{\ell}$. 
	As a first observation, note that
	in this sequence
	no two agents who have utility $-1000t$ for each other can ever be
	in the same coalition because they would never join each other.
	Second, as we start from the singleton partition, the resent an agent has build up for the
	other agents after $\ell$ deviations sums up to at most $\frac{\ell}{2}\leq 5t$.

	Next, we show that,
	in $\pi^{\ell}$, no set agent is in a joint coalition with~$p$. 
	For the sake of a contradiction,
	assume that $p\in \pi^{\ell}(S)$ for some $S\in\mathcal{S}$.
	By the first observation, this means that there is no other set agent,
	no element agent, and no filling agent in $\pi^{\ell}(S)$.
	Hence $\pi^{\ell}(S)\subseteq \{p,S\}\cup Q$.
	Let $p_{i,j}\in Q$ be a penalizing agent who has not 
	performed any deviation, yet. 
	This agent has to exist because, otherwise, we would have needed at least $10t+1$ deviations to reach $\pi^{\ell}$.
	Then, $p_{i,j}$
	has a utility of at least 
	$u_{p_{i,j}}^0(S)+u_{p_{i,j}}^0(p)-5t= 55t$ 
	for $\pi^{\ell}(S)$
	and thus has a possible NS deviation to~$\pi^{\ell}(S)$. 
	(As $p_{i,j}$ has not deviated yet, 
	she can have a utility of at most $40t$ for $\pi^{\ell}(p_{i,j})$ 
	if $p_{i,k}$ with $k\in [2], k\neq j$ has joined her.)
	This NS deviation is also an IS deviation as, initially, 
	no agent from~$\pi^{\ell}(S)$ has a negative utility for $p_{i,j}$.
	This is a contradiction to $\pi^{\ell}$ being~IS.
	So, $\pi^{\ell}(p)\subseteq \{p\}\cup Q$.
	
	As $\pi^{\ell}$ is IS, no set agent $S$ has any possible IS deviation
	from $\pi^{\ell}(S)$ to $\pi^{\ell}(p)$.
	As $p$ and all other penalizing agents would always allow $S$ to join her (note that they cannot have decreased their utility by more than $5t$ for $S$ until step $\ell$), $S$ has no NS deviation to join $\pi^{\ell}(p)$.
	Since 
	(a) the resent that $S$ has build up for any agent sums up to at most $-5t$,
	(b) $S$ initially has utilities zero for all penalizing agents $p_{i,j}\in \pi^{\ell}(p)$, 
	(c) $S$ initially has utility $60t$ for $p$, and
	(d) $S$ does not want to move to $p$,
	it follows that
	$u^{\ell}_S(\pi^{\ell}(S))\geq 60t -5t=55t$.
	
	Therefore, $\pi^{\ell}(S)$ contains at least one filling agent or three element agents $x$ with $x\in S$.
	Since this holds for all set agents
	and no two set agents are in one coalition, 
	the existence of an exact cover
	(namely $\mathcal{S'}=\{S\in\mathcal{S}\mid S \textnormal{ is together with three element agents}\}$)
	is implied.
	
	By the same proof as above,
	there is an exact cover of~$U$
	exactly if there is a sequence of at most~$k$ NS deviations
	leading to an NS partition for resentful agents.
	
	For appreciative agents, we can use the same construction.
	The proof of correctness is very similar.
\end{proof}

\section{Conclusion and Future Directions}\label{sec:discussion}
We have initiated the study of hedonic games with time-dependent utility functions, which are influenced by deviations in the past. 
In this course, we have studied a dynamic model of hedonic games by considering sequences of partitions evolving through sequences of deviations.
Whenever a deviation is performed, some of the agents influenced by this deviation modify their utility.
In particular, we have considered resentful agents, which decrease their utility for agents who abandon them,
 and appreciative agents, which increase their utility for agents who join them.
 
In our theoretical analysis, we have investigated whether the resentful or appreciative perception of other agents is sufficient to guarantee convergence for dynamics based on various deviation types.
To state our results in broad generality, we have proposed axioms for utility aggregation.
These consider the treatment of friends and enemies.
All of our axioms are satisfied for additively separable aggregation of utilities, while modified fractional utility aggregation satisfies all axioms except aversion to enemies.

For resentful agents, we have shown the convergence of dynamics based on all types of group and single-agent deviations under fairly general conditions. 
However, some of these results also need the individual rationality of deviations.
By contrast, for appreciative agents, we can only guarantee convergence of the CNS dynamics in ASHGs, while we have found several possibilities for cycling for the other types of dynamics.
Resent and appreciation do not need to be expressed by the agents affected by a deviation, but they may also affect the deviator herself.
We have seen that deviator-resent can be a strong force for stability over time.

Apart from our theoretical analysis, we have also studied the effects of resent and appreciation by means of simulations.
This gives insight in the actual running time of the dynamics as well as in the structure of the partitions obtained after the convergence of the dynamics.
We have considered two utility models, namely utilities selected uniformly at random and according to Gaussian distributions.
It turned out that the produced outcomes are fairly degenerate under uniformly random utilities.
For resentful agents, the outcome is usually the singleton partition, while the outcome consists of one large hub coalition and a few smaller coalitions for appreciative agents.
A compromise and a seemingly more realistic result is obtained if agents are affected by both resent and appreciation.
Then, the outcomes consist of several medium-size coalitions, and seem to be more related to the original base game.
By contrast, under Gaussian utilities, the differences of the outcomes for different agent types are much less profound, which suggests that results by simulations may highly depend on the creation of the random games.

We complement the simulation results on the running time by a theoretical consideration of the problem of convergence within a given time limit.
There, we obtain hardness results for dynamics in ASHGs based on any stability concept considered in this paper.
Notably, these hardness results not only hold for resentful and appreciative agents but also for for normal agents (with time-independent utilities).

Based on our work, one research direction is to consider other effects that could affect agents' valuations over time and potentially contribute to additional convergence results.
Apart from this, we have also posed several specific open questions throughout the paper (even showing the equivalence of some of them in \Cref{co:equivalence}).
In particular, it remains open, whether we can extend our convergence results for resentful agents to the modified fractional aggregation without the restriction to individually rational deviations.
There, we only know that we may cycle for general NS dynamics.
Moreover, complementing our simulations, it would be interesting to theoretically analyze the effects of combining resent and appreciation. 
A concrete open question here is whether CS dynamics are guaranteed to converge, which is the case for resentful agents but not for appreciative agents.
Finally, another path for future research is to consider fastest convergence in other classes of hedonic games apart from ASHGs.

\section*{Acknowledgments}
This work was partially supported by the Deutsche Forschungsgemeinschaft under grants NI 369/22, BR 2312/11-2, and BR 2312/12-1, and the NRW project \emph{Online Participation}.
We would like to thank the participants of the Dagstuhl seminar 21331 on \emph{Coalition Formation Games}, and especially Florian Brandl, Gr\'egory Bonnet, Edith Elkind, Bettina Klaus, Seckin \"Ozbilen, and Sanjukta Roy, for fruitful discussions.

\appendix
\section*{Appendix: Missing Proofs}

In the appendix, we provide all missing proofs.

\section{Additional Material for \Cref{se:prelims}}\label{app:prelims}

In this section, we provide the proof for our lemma considering infinite sequences of partitions.

\infinite*
\begin{proof}
	We simultaneously prove the part of the lemma about single-agent deviations and group deviations by showing a stronger lemma that implies both. To this end, we consider labeled transitions. A \emph{labeled sequence} with respect to $\mu$ is a sequence $a = (a^t)_{t\ge 0}$ together with a (labeling) function $\mu \colon \mathbb N_{\ge 0} \to L$. Given a labeled sequence $a = (a^t)_{t\ge 0}$ with respect to $\mu$, a \emph{labeled transition} of $a$ with respect to $L$ is a tuple $(l,a_1,a_2)$ such that there exists $t\ge 0$ with $a^t = a_1$, $a^{t+1} = a_2$, and $\mu(t) = l$. 
	Then, \Cref{lem:infoccurence} follows from the next claim by interpreting the deviating single agent or group of agents as labels of a transition between partitions.
	
	\begin{claim}
		Let $\ag$ be a non-empty and finite ground set and $L$ be a non-empty and finite set of labels. Let $a = (a^t)_{t\ge 0}$ be an infinite labeled sequence over $A$ with respect to $\mu$ where $a^t\in \ag$
		for all $t\ge 0$.  Then, there exists a $t_0\ge 0$ such that every labeled transition of $a$ with respect to $L$ performed at some time $t\ge t_0$ occurs infinitely often.
	\end{claim}

	\renewcommand\qedsymbol{$\vartriangleleft$} 
	
	\begin{proof}
		Let $a = (a^t)_{t\ge 0}$ be an infinite labeled sequence with respect to $\mu$ over finite and non-empty ground sets and label sets. Given $a_1,a_2\in \ag$
		and $l\in L$, define $t(l,a_1,a_2) = \sup\{t\ge 0\colon a^t = a_1, a^{t+1} = a_2, \mu(t) = l\}$, where we set $\sup \emptyset = 0$. In other words, $t(l,a_1,a_2)$ defines the last time step where the labeled transition $(l,a_1,a_2)$ occurs. We define $t_0 = \max\{t(l,a_1,a_2)\colon l\in L, a_1,a_2\in \ag, t(l,a_1,a_2) < \infty\}$. By construction, every labeled transition performed after $t_0$ must occur infinitely often.
	\end{proof}
	
	This completes the proof of the lemma.
	\renewcommand\qedsymbol{$\square$}
\end{proof}

\section{Additional Material for \Cref{sec:apprec}}

\resentappreciation*
\begin{proof}
	We only show how the first statement implies the second statement. The reverse implication is completely analogous.
	
	Assume that every agent $i\in N$ is resentful and uses the {\caf} $\mathit{AS}_i$. Assume that there exists an infinite and periodic sequence $(\pi^t)_{t\ge 0}$ of NS deviations.
	Assume that $\pi^t$ evolves from $\pi^{t-1}$ through an NS deviation of agent~$d^t$ for all $t\ge 1$. 
	By periodicity, there exist $t_0\in \mathbb N$ and $p\in \mathbb N$ such that, for all $k\in \mathbb N_0$ and $l\in \{0,\dots, p-1\}$, it holds that
	$\pi^{t_0+kp+l} = \pi^{t_0+l}$.
	We define an infinite and periodic sequence for appreciative agents. Essentially, this sequence reverts the order of one (periodic) segment of $(\pi^t)_{t\ge 0}$ and appends it indefinitely.
	
	Given $k\in \mathbb N_0$ and $l\in \{0,\dots, p-1\}$, define $\sigma^{kp+l} = \pi^{t_0+p-l}$. Clearly, this defines an infinite and periodic sequence of partitions $(\sigma^t)_{t\ge 0}$ induced by single-agent deviations. In particular, partition $\sigma^{kp+l+1}$ evolves from partition $\sigma^{kp+l}$ by a single-agent deviation of agent $d^{t_0+p-l}$.
	Moreover, define initial utilities by $\bar u^0_i(j) = - u^{t_0}_i(j)$. We assume that these utilities are updated in $(\sigma^t)_{t\ge 0}$ according to appreciation.
	
	Given a pair of agents $i,j$, let $r_i(j) = |\{l\in \{1,\dots, p\}\colon d^{t_0+l} = j, i\in \pi^{t_0+l}(j)\}|$, i.e., the resent that $i$ builds for $j$ during one period. Also, given a pair of agents $i,j$ and $s\in \{1,\dots, p\}$, let $r_i^s(j) = |\{l\in \{1,\dots, s\}\colon d^{t_0+l} = j, i\in \pi^{t_0+l}(j)\}|$, i.e., the resent that $i$ builds for $j$ until step $s$ of a period.
	
	We have to show that the single-agent deviations are even NS deviations.
	Therefore, fix $k\in \mathbb N_0$, $l\in  \{0,\dots, p-1\}$, and define $\sigma = \sigma^{kp+l}$, $\sigma' = \sigma^{kp+l+1}$, and $d = d^{t_0+p-l}$. We first claim that
	\begin{align}\label{eq:builtresent}
	\sum_{j\in \sigma'(d)}r_d(j) \ge \sum_{j\in \sigma(d)}r_d(j)\text{,} 
	\end{align}
	i.e., the appreciation of a deviator build during one period for an abandoned coalition is bounded by the appreciation build for the joined coalition. 
	
	Suppose the contrary. Clearly, the resent that agents aggregate in one iteration of the cycle $(\pi^{t_0},\dots, \pi^{t_0+p})$ is the same as the appreciation that agents aggregate in a cycle $(\sigma^0,\dots, \sigma^p)$. 
	
	Then, it holds for all $m\in \mathbb N_0$ that
	
	\begin{align*}
	& 0 < u_d^{t_0 + mp-l-1}(\pi^{t_0+p-l})-u_d^{t_0 + mp-l-1}(\pi^{t_0+p-l-1})\\
	& = u_d^{t_0 + p -l -1}(\pi^{t_0+p-l})-u_d^{t_0 + p -l -1}(\pi^{t_0+p-l-1})\\ & - (m-1) \left(\sum_{j\in \pi^{t_0+p-l}(d)}r_d(j) - \sum_{j\in \pi^{t_0+p-l-1}(d)}r_d(j)\right)\\
	& = u_d^{t_0 + p -l-1}(\pi^{t_0+p-l})-u_d^{t_0 + p -l -1}(\pi^{t_0+p-l-1}) \\ &- (m-1) \left(\sum_{j\in \sigma(d)}r_d(j) - \sum_{j\in \sigma'(d)}r_d(j)\right)\text{.}
	\end{align*}
	
	Therefore, 
	\begin{align*}
	& (m-1) \left(\sum_{j\in \sigma(d)}r_d(j) - \sum_{j\in \sigma'(d)}r_d(j)\right) \\ &  < u_d^{t_0 + p -l - 1}(\pi^{t_0+p-l})-u_d^{t_0 + p -l - 1}(\pi^{t_0+p-l-1})\text{.} 
	\end{align*}
	
	However, the left hand side of this equation is unbounded while the right hand side is fixed. Hence, this cannot be true for all $m\in \mathbb N_0$, a contradiction. We have thus established \Cref{eq:builtresent}.
	
	We can compute that the deviation is indeed a NS deviation in every period:
	
	\begin{align*}
	& \bar u_d^{kp+l}(\sigma') - \bar u_d^{kp+l}(\sigma)\\
	& = \bar u_d^{l}(\sigma') - \bar u_d^{l}(\sigma)\\
	& + k \left(\sum_{j\in \sigma'(d)}r_d(j) - \sum_{j\in \sigma(d)}r_d(j)\right)\\
	& \overset{\mathit{\eqref{eq:builtresent}}}{\ge} \bar u_d^{l}(\sigma') - \bar u_d^{l}(\sigma)\\
	& = \bar u_d^{0}(\sigma') + \sum_{j\in \sigma'(d)}(r_d(j)-r_d^{p-l}(j))\\& - \bar u_d^{0}(\sigma) - \sum_{j\in \sigma(d)}(r_d(j)-r_d^{p-l}(j))\\
	& = - u_d^{t_0}(\pi^{t_0+p-l-1}) + \sum_{j\in \pi^{t_0+p-l-1}(d)}(r_d(j)-r_d^{p-l}(j))\\& + u_d^{t_0}(\pi^{t_0+p-l}) - \sum_{j\in \pi^{t_0+p-l}(d)}(r_d(j)-r_d^{p-l}(j))\\
	& = - u_d^{t_0+p-l}(\pi^{t_0+p-l-1}) + \sum_{j\in \pi^{t_0+p-l-1}(d)}r_d(j)\\& + u_d^{t_0+p-l}(\pi^{t_0+p-l}) - \sum_{j\in \pi^{t_0+p-l}(d)}r_d(j)\\
	& \overset{\mathit{\eqref{eq:builtresent}}}{\ge} u_d^{t_0+p-l}(\pi^{t_0+p-l}) - u_d^{t_0+p-l}(\pi^{t_0+p-l-1}) > 0\text{.}
	\end{align*}
	
	The strict inequality in the end holds because the deviation performed by $d$ in the sequence $(\pi^t)_{t\ge 0}$ is a NS deviation. 
	Hence, the sequence $(\sigma^t)_{t\ge 0}$ evolves through NS deviations.
\end{proof}

\begin{table}
		\caption{Example from \Cref{th:mfhg-apprec}. Utilities after $x$ cycles of deviations. Each row depicts the utility that one agent has for the other agents.} \label{table:mfhg-apprec}
		\centering
		\begin{tabular}{ c|c|c|c|c|c|c } 
			& $a$ & $a'$ & $b$ & $b'$ & $c$ & $c'$ \\ \hline
			$a$ & $-$ & $110+x$ & $120+x$ & $-100+x$ & $130+x$ & $-100+x$ \\ \hline
			$a'$ & $20+x$ & $-$ & $100+x$ & $10+x$ & $100+x$ & $30+x$ \\ \hline
			$b$ & $130+x$ & $-100+x$ & $-$ & $110+x$ & $120+x$ & $-100+x$ \\ \hline
			$b'$ & $100+x$ & $30+x$ & $20+x$ & $-$ & $100+x$ & $10+x$ \\ \hline
			$c$ & $120+x$ & $-100+x$ & $130+x$ & $-100+x$ & $-$ & $110+x$ \\ \hline
			$c'$ & $100+x$ & $10+x$ & $100+x$ & $30+x$ & $20+x$ & $-$
		\end{tabular}
	\end{table}

\mfhgapprec*
\begin{proof}
	We now describe an involved example of an MFHG together with an infinite periodic sequence of NS deviations for appreciative agents. 
	Our example is very similar to the one presented for resentful agents in \Cref{th:mfhg-resent}. 
	In fact, following the general idea of \Cref{thm:NS-resent-vs-appreciation-ASHG}, we reverse the deviation sequence and appropriately adjust the initial utilities.
	As in \Cref{th:mfhg-resent},  each agent joins each other agent exactly once in each cycle. 
	This establishes the same invariant: if a deviating agent  prefers a joined coalition $C_1$ to an abandoned coalition $C_2$
	before (and during) the \emph{first} execution of the cycle, then it still prefers $C_1$ to $C_2$ before (and during) \emph{each} execution of the cycle.
	
	Consider the game with agent set $N = \{a,a',b,b',c,c'\}$ and utilities as depicted in \Cref{table:mfhg-apprec}. The initial utilities result from setting $x = 0$ in \Cref{table:mfhg-apprec}. 
	We now present an infinite sequence $(\pi^t)_{t\ge 0}$ of partitions, always consisting of three coalitions. 
	For each partition, we refer to the first listed coalition as $C_1$, to the second as $C_2$, and the third as $C_3$. 
	For the sake of clarity, for each partition, we also specify which agent deviates to which coalition in the next step. 
	Specifically, for $n\geq 0$, we have
	\begin{itemize}
		\item $\pi^{18n+1} = \{\{b'\},\{a,a'\},\{b,c,c'\}\}$ with agent $b$ deviating to $C_2$,
		\item $\pi^{18n+2} = \{\{b'\},\{b,a,a'\},\{c,c'\}\}$ with agent $c'$ deviating to $C_1$,
		\item $\pi^{18n+3} = \{\{b',c'\},\{b,a,a'\},\{c\}\}$ with agent $c'$ deviating to $C_2$,
		\item $\pi^{18n+4} = \{\{b'\},\{b,a,a',c'\},\{c\}\}$ with agent $b$ deviating to $C_1$,
		\item $\pi^{18n+5} = \{\{b',b\},\{a,a',c'\},\{c\}\}$ with agent $a$ deviating to $C_1$,
		\item $\pi^{18n+6} = \{\{b',b,a\},\{a',c'\},\{c\}\}$ with agent $c'$ deviating to $C_3$,
		\item $\pi^{18n+7} = \{\{b',b,a\},\{a'\},\{c,c'\}\}$ with agent $a$ deviating to $C_3$,
		\item $\pi^{18n+8} = \{\{b',b\},\{a'\},\{c,c',a\}\}$ with agent $b'$ deviating to $C_2$,
		\item $\pi^{18n+9} = \{\{b\},\{a',b'\},\{c,c',a\}\}$ with agent $b'$ deviating to $C_3$,
		\item $\pi^{18n+10} = \{\{b\},\{a'\},\{c,c',a,b'\}\}$ with agent $a$ deviating to $C_2$,
		\item $\pi^{18n+11} = \{\{b\},\{a',a\},\{c,c',b'\}\}$ with agent $c$ deviating to $C_2$,
		\item $\pi^{18n+12} = \{\{b\},\{a',a,c\},\{c',b'\}\}$ with agent $b'$ deviating to $C_1$,
		\item $\pi^{18n+13} = \{\{b,b'\},\{a',a,c\},\{c'\}\}$ with agent $c$ deviating to $C_1$,
		\item $\pi^{18n+14} = \{\{b,b',c\},\{a',a\},\{c'\}\}$ with agent $a'$ deviating to $C_3$,
		\item $\pi^{18n+15} = \{\{b,b',c\},\{a\},\{c',a'\}\}$ with agent $a'$ deviating to $C_1$,
		\item $\pi^{18n+16} = \{\{b,b',c,a'\},\{a\},\{c'\}\}$ with agent $c$ deviating to $C_3$,
		\item $\pi^{18n+17} = \{\{b,b',a'\},\{a\},\{c',c\}\}$ with agent $b$ deviating to $C_3$,
		\item $\pi^{18n+18} = \{\{b',a'\},\{a\},\{c',c,b\}\}$ with agent $a'$ deviating to $C_2$.
	\end{itemize}

	Then, it is possible to verify that for $k\ge 1$, $\pi^{k-1}$ leads to $\pi^k$ by means of an NS deviation. Hence, we have presented an MFHG with an infinite sequence of NS deviations for appreciative agents. 	
\end{proof}

\section{Additional Material for \Cref{app:shortest}}

The goal of this section is to complete the proof of \Cref{thm:sequence}.
Since we have already considered NS and IS in \Cref{thm:sequence-NS-IS}, it remains to consider CNS and CS.
We split the remaining proof into two more lemmas.

\begin{lemma}\label{thm:sequence-CNS}
	{\sc resentful-} and {\sc appreciative-ashg-CNS-Sequence} are $\NP$-hard. 
\end{lemma}
\begin{proof}
	We first show the statement for resentful agents. 
	We again reduce from \textsc{RX3C}. 
	
	\textbf{Construction.}
	From an instance of \textsc{RX3C}, 
	given by a universe $U=\{x_1,\dots, x_{3t}\}$ and 
	a family $\mathcal{S}=\{S_1,\dots , S_{3t}\}$ of $3$-subsets of $U$,
	we create an additively separable hedonic game as follows. 
	We create the set of agents 
	$N=U\cup \mathcal{S} \cup F \cup Q$
	with element agents $U$,
	set agents~$\mathcal{S}$,
	filling agents $F=\{f_1,\ldots,f_{2t}\}$, and
	penalizing agents $Q=\{p_0,\ldots,p_4\}$.
	The initial utilities of the agents
	are illustrated in Figure~\ref{fig:sequenceCNS}:

	\begin{itemize}
		\item Each element agent has utility $1$ for each other element agent.
		\item Each set agent $S\in \mathcal{S}$ has 
		utility $20t$ for the three element agents $x\in S$,  
		utility $60t$ for each filling agent, 
		utility~$60t$ for~$p_0$, and
		utility zero for $p_1$.
		\item $p_0$ has 
		utility $10t$ for each set agent and
		utility zero for~$p_1$.
		\item $p_1$ has 
		utility $1$ for $p_0$.
		\item $p_2$, $p_3$, and $p_4$ have utility $200t$ for $p_1$.
		\item $p_2$ has 
		utility $-100t$ for $p_3$ and
		utility $-400t$ for $p_4$.
		\item $p_3$ has
		utility $-100t$ for $p_4$ and
		utility $-400t$ for $p_2$.
		\item $p_4$ has
		utility $-100t$ for $p_2$ and
		utility $-400t$ for $p_3$.
		\item All not explicitly mentioned utilities are $-1000t$.
	\end{itemize}
	The penalizing gadget consisting of agents~$p_1,\ldots,p_4$
	is a variant of an ASHG by
	\citet[Example~2]{sun-dim:j:on-myopic-stability-concepts-for-hedonic-games}
	for which no CNS partition exists (in normal ASHGs).
	
	We set $k$ to $5t+1$ and
	show that there is an exact cover of~$U$
	if and only if there is a sequence starting from the singleton partition of at most $k$ CNS deviations
	leading to an CNS partition, where agents are resentful.
		\begin{figure}
		\centering
		\begin{tikzpicture}[
		element/.style={shape=circle,draw, fill=white}
		]
		\pgfmathsetmacro\xdist{2}
		\pgfmathsetmacro\ydist{-1}
		
		\node (x1) at (0,0) {${x_1}$};
		\node (dots1) at (0,\ydist) {$\vdots$};;
		\node (x3t) at (0,2*\ydist) {${x_{3t}}$};
		
		\node (s1) at (\xdist,0) {$S_1$};
		\node (dots2) at (\xdist,\ydist) {$\vdots$};
		\node (s3t) at (\xdist,2*\ydist) {$S_{3t}$};
		
		\node (f1) at (2*\xdist,0.4*\ydist) {$f_{1}$};
		\node (dots3) at (2*\xdist,\ydist) {$\vdots$};
		\node (f2t) at (2*\xdist,1.6*\ydist) {$f_{2t}$};
		
		\node (p0) at (\xdist,4*\ydist) {$p_0$};
		\node (p1) at (2*\xdist,4*\ydist) {$p_1$};
		\node (p2) at (3*\xdist,2*\ydist) {$p_2$};
		\node (p3) at (4*\xdist,4*\ydist) {$p_3$};
		\node (p4) at (3*\xdist,6*\ydist) {$p_4$};

		\node (E) [draw, dashed, rounded corners=3pt, fit={(x1) (x3t)}] {};
		\node at (x3t) [below=0.5cm] {$1$};

		\draw[->] (dots2) edge  
		node[fill=white,anchor=center, pos=0.5, inner sep =2pt, above] {\footnotesize $20t$} 
		node[fill=white,anchor=center, pos=0.5, inner sep =2pt, below] {\footnotesize $x\in S$}
		(dots1);
		\draw[->] (dots2) edge  
		node[fill=white,anchor=center, pos=0.5, inner sep =2pt] {\footnotesize $60t$} (dots3);

		\draw[->, bend right] (s3t) edge  
		node[fill=white,anchor=center, pos=0.5, inner sep =2pt] 
		{\footnotesize $60t$} (p0);
		\draw[->, bend right=20] (p0) edge  
		node[fill=white,anchor=center, pos=0.5, inner sep =2pt] 
		{\footnotesize $10t$} (s3t);
		
		\draw[->, bend right=20] (p0) edge  
		node[fill=white,anchor=center, pos=0.5, inner sep =2pt] 
		{\footnotesize $0$} (p1);
		\draw[->, bend right=20] (p1) edge  
		node[fill=white,anchor=center, pos=0.5, inner sep =2pt] 
		{\footnotesize $1$} (p0);
		
		\draw[->, bend left=20] (s3t) edge  
		node[fill=white,anchor=center, pos=0.5, inner sep =2pt] {\footnotesize $0$} (p1);
		
		\draw[->] (p2) edge  
		node[fill=white,anchor=center, pos=0.5, inner sep =2pt] 
		{\footnotesize $200t$} (p1);
		\draw[->] (p3) edge  
		node[fill=white,anchor=center, pos=0.75, inner sep =2pt] 
		{\footnotesize $200t$} (p1);
		\draw[->] (p4) edge  
		node[fill=white,anchor=center, pos=0.5, inner sep =2pt] 
		{\footnotesize $200t$} (p1);
		
		\draw[->, bend right=15] (p4) edge  
		node[fill=white,anchor=center, pos=0.4, inner sep =2pt] 
		{\footnotesize $-400t$} (p3);
		\draw[->, bend right=15] (p3) edge  
		node[fill=white,anchor=center, pos=0.4, inner sep =2pt] 
		{\footnotesize $-100t$} (p4);
		
		\draw[->, bend right=15] (p2) edge  
		node[fill=white,anchor=center, pos=0.7, inner sep =2pt] 
		{\footnotesize $-400t$} (p4);
		\draw[->, bend right=15] (p4) edge  
		node[fill=white,anchor=center, pos=0.6, inner sep =2pt] 
		{\footnotesize $-100t$} (p2);
		
		\draw[->, bend right=15] (p3) edge  
		node[fill=white,anchor=center, pos=0.6, inner sep =2pt] 
		{\footnotesize $-400t$} (p2);
		\draw[->, bend right=15] (p2) edge  
		node[fill=white,anchor=center, pos=0.6, inner sep =2pt] 
		{\footnotesize $-100t$} (p3);
		
		\end{tikzpicture}
		\caption{Additively separable hedonic games in  \Cref{thm:sequence-CNS}. All omitted utilities are $-1000t$.\label{fig:sequenceCNS}}
	\end{figure}
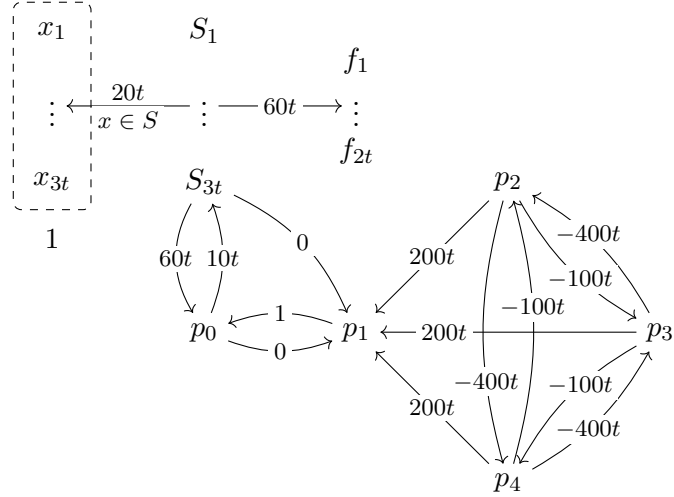 
	{\bfseries ($\Rightarrow$)}
	Given an exact cover $\mathcal{S}'\subseteq \mathcal{S}$, 
	consider the following sequence of $k=5t+1$ deviations.
	First, for each set $S\in \mathcal{S}'$ with $S=\{x_i,x_j,x_l\}$,
	$x_i$ and $x_j$ join $x_l$
	and afterwards the set agent $S$ joins $\{x_i,x_j,x_l\}$
	($3t$ deviations).
	Second, all set agents $S$ with $S\notin \mathcal{S}'$ join 
	a filling agent each
	($2t$ deviations).
	Last, $p_1$ joins $p_0$.
	
	All these deviations are NS deviations as 
	each of these deviations is performed from a singleton 
	to a coalition with positive utility for the deviator.
	Moreover, all these NS deviations are CNS deviations
	because they are performed from singletons;
	so there are no other agents who could veto the deviation.
	
	Next, we show that the resulting partition $\pi$ is CNS.
	First, each set agent has utility $60t$ for $\pi(S)$.
	Since every other coalition either contains an agent for which she has utility $-1000t$
	or is the coalition $\{p_0,p_1\}$, for which she has utility $60t$,
	she does not want to deviate.
	Second, each element and filling agent do not have any CNS deviation
	because they have a set agent in their coalition who would
	veto their deviation.
	Third, $p_1$ would veto a deviation of $p_0$, so $p_0$ has no CNS deviation.
	Fourth, $p_1$ does not want to deviate to another coalition as 
	she has utility $1$ for $\pi(p_1)=\{p_0,p_1\}$ 
	and every other coalition in $\pi$ contains an agent for which 
	she has utility $-1000t$.
	Last, $p_2$, $p_3$, and $p_4$ do not want to deviate as 
	they only have a positive utility (of $200t$) for agent $p_1$ but
	$p_1$ is together with $p_0$ for which they have utility $-1000t$.
	
	{\bfseries ($\Leftarrow$)}
	Assume that we have a sequence of $\ell\leq k=5t+1$ CNS deviations 
	leading from the singleton partition $\pi^0$ 
	to a CNS partition $\pi^{\ell}$. 
	In this sequence,
	no two agents who mutually have utility $-1000t$ for each other can ever be
	in the same coalition because they would never join each other.
	Furthermore, as we start from the singleton partition, the resent an agent has accumulated for the
	other agents after $\ell$ deviations sums up to at most $\frac{\ell}{2}\leq 3t$.
	
	We will now show that no set agent is together with~$p_0$ in $\pi^{\ell}$.
	For the sake of a contradiction,
	assume that $S\in \pi^{\ell}(p_0)$ for some $S\in \mathcal{S}$.
	By the first observation, this implies that
	$\pi^{\ell}(p_0)\subseteq \{S,p_0,p_1\}$.
	It further holds that $p_1\notin \pi^{\ell}(p_0)$
	because, otherwise, $p_1$ would have a CNS deviation
	to a singleton coalition.
	Hence, $\pi^{\ell}(p_0)= \{S,p_0\}$.
	Next, 
	consider coalition $\pi^{\ell}(p_1)$.
	By the first observation, we have 
	$\pi^{\ell}(p_1)\subseteq \mathcal{S}\cup Q$.
	Furthermore, there is no set agent in $\pi^{\ell}(p_1)$:
	Indeed, assume for contradiction that there is some set agent $S'$ in $\pi^{\ell}(p_1)$.
	Then, by the first observation, neither $p_2$, $p_3$, nor $p_4$ 
	are in $\pi^{\ell}(p_1)$.
	Furthermore, we already now that $p_0$ is not in $\pi^{\ell}(p_1)$.
	It follows that $\pi^{\ell}(p_1)=\{S',p_1\}$ which means that
	$p_1$ has a CNS deviation to a singleton coalition, a contradiction.
	Hence, $\pi^{\ell}(p_1)\subseteq \{p_1,\ldots,p_4\}$.
	For $i\in \{2,3,4\}$, it also holds by the first observation 
	that
	$\pi^{\ell}(p_i)\subseteq \{p_1,\ldots,p_4\}$.
	We will now show that $\{p_1,\ldots,p_4\}$ can not be divided into a CNS partition
	\citep[cf.][Example~2]{sun-dim:j:on-myopic-stability-concepts-for-hedonic-games}.
	Note that the arguments also hold when lowering any utilities among this agents by at most $3t$ and therefore $\pi^{\ell}$ is no CNS partition,
	which is a contradiction. 
	
	First, none of $p_2$, $p_3$, and $p_4$ can be in the same coalition 
	because the agent with a initial utility of $-400t$ for another agent in her coalition
	would have a CNS deviation to a singleton.
	Hence, the remaining possible partitions of $\{p_1,\ldots,p_4\}$ are
	$\pi'=\{\{p_1\},\{p_2\},\{p_3\},\{p_4\}\}$,
	$\pi''=\{\{p_1,p_2\},\{p_3\},\{p_4\}\}$,
	$\pi'''=\{\{p_2\},\{p_1,p_3\},\{p_4\}\}$, and
	$\pi''''=\{\{p_2\},\{p_3\},\{p_1,p_4\}\}$.
	But all these partitions are not CNS:
	$\pi'$ is not stable because $p_2$ has a CNS deviation to $\{p_1\}$;
	$\pi''$ is not stable because $p_4$ has a CNS deviation to $\{p_1,p_2\}$;
	$\pi'''$ is not stable because $p_2$ has a CNS deviation to $\{p_1,p_3\}$;
	$\pi''''$ is not stable because $p_3$ has a CNS deviation to $\{p_1,p_4\}$.
	Thus, we have shown that $\pi^{\ell}$ is not CNS. This is a contradiction
	and therefore the initial assumption was wrong, i.e., no set agent is together with $p_0$ in $\pi^{\ell}$.
	
	Since $\pi^{\ell}$ is CNS, no set agent has a CNS deviation to $p_0$.
	Note that $p_0$ is the only agent who has positive utilities for the set agents.
	So, for any set agent $S\in \mathcal{S}$,
	there is no agent in $\pi^{\ell}(S)$ who would veto a deviation of $S$.
	Therefore, no set agent has an NS deviation to $p_0$.
	Furthermore, we know that $\pi^{\ell}(p_0)\subseteq\{p_0,p_1\}$
	because of the first observation and because no set agent is with $p_0$.
	So, for every set agent $S\in \mathcal{S}$, 
	we have
	$u_{S}^{\ell}(\pi^{\ell})\geq
	u_{S}^{\ell}(\pi^{\ell}(p_0)\cup\{S\})\geq 60t-3t=57t$ 
	(where $60t$ is $S$'s initial utility for $p_0$ and
	$-3t$ is a bound for the resent that $S$ might have build for $p_0$ or $p_1$).
	Therefore, $\pi^{\ell}(S)$ contains at least one filling agent or three element agents $x$ with $x\in S$.
	Since no two set agents are in the same coalition, 
	this implies the existence of an exact cover for $U$.
	
	For appreciative agents, we can use almost the same construction.
	We just set $p_0$'s initial utility for $p_1$ to $-1$
	while all other utilities stay the same.
	Then, the proof of correctness is very similar to the proof for resentful agents.
\end{proof}

\begin{lemma}\label{thm:sequence-CS}
	{\sc resentful-} and {\sc appreciative-ashg-CS-Sequence} are $\NP$-hard.
\end{lemma}
\begin{proof}
	We first show the statement for resentful agents 
	and reduce from \textsc{RX3C}.
	
	\textbf{Construction.}
	Let an instance of \textsc{RX3C} $(U,\mathcal{S})$ with 
	a universe $U=\{x_1,\dots, x_{3t}\}$ and 
	a family $\mathcal{S}=\{S_1,\dots , S_{3t}\}$ of $3$-subsets of $U$
	be given.
	We set $k=3t+2$ and
	create the following additively separable hedonic game. 
	The set of agents is
	$N=U\cup \mathcal{S} \cup F \cup Q$
	with element agents $U$,
	set agents~$\mathcal{S}$,
	filling agents $F=\{f_1,\ldots,f_{2t}\}$, and
	penalizing agents $Q=\{p_0,\ldots,p_6\}$.
	The initial utilities of the agents
	are illustrated in Figure~\ref{fig:sequence-CS}:
	\begin{figure}
		\centering
		\begin{tikzpicture}[
		element/.style={shape=circle,draw, fill=white}
		]
		\pgfmathsetmacro\xdist{2}
		\pgfmathsetmacro\ydist{-1}
		
		\node (x1) at (0,0) {${x_1}$};
		\node (dots1) at (0,\ydist) {$\vdots$};;
		\node (x3t) at (0,2*\ydist) {${x_{3t}}$};
		\node (xphan) at (0,2.1*\ydist) {\phantom{${x_{3t}}$}};

		\node (S1) at (\xdist,0) {$S_1$};
		\node (dots2) at (\xdist,\ydist) {$\vdots$};
		\node (S3t) at (\xdist,2*\ydist) {$S_{3t}$};
		\node (Sphan) at (\xdist,2.2*\ydist) {\phantom{$S_{3t}$}};
		\node (Sphan2) at (0.9*\xdist,2*\ydist) {\phantom{$S_{3t}$}};
		
		\node (f1phan) at (1.8*\xdist,0.2*\ydist) {\phantom{$f_{1}$}};
		\node (f1) at (2*\xdist,0.4*\ydist) {$f_{1}$};
		\node (dots3) at (2*\xdist,\ydist) {$\vdots$};
		\node (f2t) at (2*\xdist,1.6*\ydist) {$f_{2t}$};
		\node (fphan) at (2*\xdist,1.8*\ydist) {\phantom{$f_{2t}$}};

		\node (p0) at (1*\xdist,3.5*\ydist) {$p_0$};
		
		\node (p1) at (2*\xdist,3.5*\ydist) {$p_{1}$};
		\node (p2) at (3*\xdist,3.5*\ydist) {$p_{2}$};
		\node (p3) at (3.5*\xdist,4.8*\ydist) {$p_{3}$};
		\node (p4) at (3*\xdist,6.1*\ydist) {$p_{4}$};
		\node (p5) at (2*\xdist,6.1*\ydist) {$p_{5}$};
		\node (p6) at (1.5*\xdist,4.8*\ydist) {$p_{6}$};
		
		\node (pphan) at (1*\xdist,5.4*\ydist) {\phantom{$p_{6}$}};

		\node (E) [draw, dashed, rounded corners=3pt, fit={(x1) (x3t)}] {};
		\node at (x1) [above=0.35cm] {\footnotesize$0$};
		
		\draw[->, bend right=15] (p0) edge
		node[fill=white,anchor=center, pos=0.5, inner sep =2pt] {\footnotesize $10t$} (p1);
		\draw[<-, bend left=15] (p0) edge
		node[fill=white,anchor=center, pos=0.5, inner sep =2pt] {\footnotesize $300t$} (p1);
		
		\draw[<->] (p1) edge
		node[fill=white,anchor=center, pos=0.5, inner sep =2pt] {\footnotesize $60t$} (p2);
		\draw[<->] (p3) edge
		node[fill=white,anchor=center, pos=0.5, inner sep =2pt] {\footnotesize $60t$} (p4);
		\draw[<->] (p5) edge
		node[fill=white,anchor=center, pos=0.5, inner sep =2pt] {\footnotesize $60t$} (p6);
		
		\draw[<->] (p2) edge
		node[fill=white,anchor=center, pos=0.5, inner sep =2pt] {\footnotesize $50t$} (p3);
		\draw[<->] (p4) edge
		node[fill=white,anchor=center, pos=0.5, inner sep =2pt] {\footnotesize $50t$} (p5);
		\draw[<->] (p6) edge
		node[fill=white,anchor=center, pos=0.5, inner sep =2pt] {\footnotesize $50t$} (p1);
		
		\draw[<->] (p1) edge
		node[fill=white,anchor=center, pos=0.5, inner sep =2pt] {\footnotesize $40t$} (p3);
		\draw[<->] (p3) edge
		node[fill=white,anchor=center, pos=0.5, inner sep =2pt] {\footnotesize $40t$} (p5);
		\draw[<->] (p5) edge
		node[fill=white,anchor=center, pos=0.5, inner sep =2pt] {\footnotesize $40t$} (p1);

		\draw[->, bend right] (S3t) edge
		node[fill=white,anchor=center, pos=0.5, inner sep =2pt] {\footnotesize $60t$} (p0);
		\draw[<-,bend left] (S3t) edge
		node[fill=white,anchor=center, pos=0.5, inner sep =2pt] {\footnotesize $20t$} (p0);

		\draw[->] (0.3,0.6*\ydist) -- (0.8*\xdist,0.6*\ydist) node[fill=white,anchor=center, pos=0.5, inner sep =2pt] {\footnotesize $1$};
		\draw[<-] (0.3,1.4*\ydist) -- (0.8*\xdist,1.4*\ydist) node[fill=white,anchor=center, pos=0.5, inner sep =2pt, above] {\footnotesize $20t$}
		node[fill=white,anchor=center, pos=0.5, inner sep =2pt, below] {\footnotesize $x\in S$};
		
		\draw[<-] (1.2*\xdist,0.8*\ydist) -- (1.7*\xdist,0.8*\ydist) node[fill=white,anchor=center, pos=0.5, inner sep =2pt] {\footnotesize $1$};
		\draw[->] (1.2*\xdist,1.2*\ydist) -- (1.7*\xdist,1.2*\ydist) node[fill=white,anchor=center, pos=0.5, inner sep =2pt] {\footnotesize $60t$};
		
		\end{tikzpicture}
		\caption{Additively separable hedonic games in  \Cref{thm:sequence-CS}. 
			All omitted utilities are $-1000t$.\label{fig:sequence-CS}}
	\end{figure}
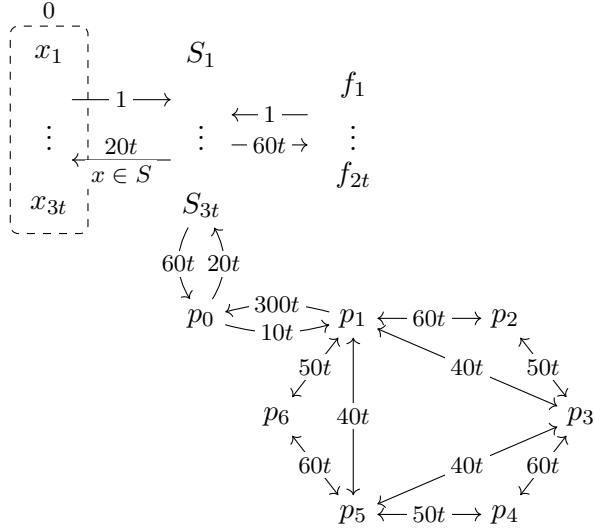 
	\begin{itemize}
		\item Each set agent $S$ has 
		utility $60t$ for each filling agent, 
		utility~$20t$ for element agents $x$ with $x\in S$, 
		and utility $60t$ for the penalizing agent $p_0$. 
		\item Each element agent 
		has utility $1$ for each set agent
		and utility $0$ for all other element agents.
		\item Each filling agent has utility $1$ for each set agent. 
		\item $p_0$ has utility $20t$ for all set agents and
		utility $10t$ for $p_1$.
		\item Agents $p_1$, $p_3$, and $p_5$ have utility $40t$ for each other.
		\item For $i\in\{1,3,5\}$, agent $p_i$ has 
		utility $60t$ for agent $p_{i+1}$ and vice versa.
		\item For $i\in\{2,4,6\}$, agent $p_i$ has 
		utility $50t$ for agent $p_{i+1}$ and vice versa (where $p_7=p_1$).
		\item All not explicitly mentioned utilities are $-1000t$. 
	\end{itemize}
	
	Note that the penalizing gadget consisting of agents~$p_1,\ldots,p_6$
	is a variant of an ASHG by
	\citet[Example~1]{ABS11c}
	for which no CS partition exists (for normal ASHGs).
	
	{\bfseries ($\Rightarrow$)}
	Given an exact cover $\mathcal{S}'\subseteq \mathcal{S}$, 
	consider the following sequence of group deviations.
	For each $S\in \mathcal{S}'$,
	$S$ and the three $x\in S$ deviate together ($t$ deviations).
	For each $S\in \mathcal{S}\setminus\mathcal{S}'$,
	$S$ deviates with one filling agent that is still in a singleton ($2t$ deviations).
	Now, $p_0$ deviates with $p_1$.
	Lastly, $p_3$, $p_4$, and $p_5$ deviate together.
	Hence, we have $3t+2$ deviations in total.
	
	All these deviations are CS deviations as 
	all these deviations are performed from singletons 
	to coalitions with positive utilities for all deviators.
	
	Next, we show that the resulting partition $\pi$ is CS.
	First note that no two agents who have utility $-1000t$ for each other
	want to deviate together.
	Each element agent $x$ 
	only has a positive utility of $1$ for all set agents and
	is in a coalition with one set agent.
	As no two set agents will ever deviate together, there is no coalition that
	$x$ could deviate with to improve her utility.
	The same holds for all filling agents.
	Each set agent $S$ has utility $60t$ for $\pi(S)$.
	As no element or filling agent wants to deviate, 
	the only remaining agent for which $S$ has positive utility is $p_0$.
	Yet, $S$ only has utility $60t$ for $p_0$ which is no improvement
	compared to $\pi(S)$.
	So no set agent has a CS deviation.
	Therefore, also $p_0$ has no CS deviation as she would only 
	like to deviate with some set agents.
	$p_1$ has utility $300t$ for $\pi(p_1)=\{p_0,p_1\}$ and
	would thus not deviate without $p_0$.
	Similarly, the remaining agents $p_2,\ldots,p_6$ have
	no CS deviation without $p_1$.
	Thus, the whole partition $\pi$ is CS.
	
	{\bfseries ($\Leftarrow$)}
	Assume that we have a sequence of $\ell\leq k=3t+2$ CS deviations 
	leading from the singleton partition $\pi^0$ to a CS partition $\pi^{\ell}$. 
	In this sequence,
	no agent is ever in a coalition 
	with an agent for which she has utility $-1000t$ 
	because she would never deviate with this agent.
	Therefore, the largest coalitions in this sequence have size at most four
	(containing one set agent and three element agents).
	It follows that
	the sum of the resent that a single agent accumulates for the other agents 
	during $\ell\leq 3t+2$ deviations 
	is bounded by $\frac{\ell}2 \cdot 3 \le \frac{3t+2}{2}\cdot 3=\frac 92 t+3$ because
	she can be left by at most three agents every second deviation.
	
	We will now show that no set agent is in a joint coalition with~$p_0$ in $\pi^{\ell}$.
	For the sake of a contradiction,
	assume that $S\in \pi^{\ell}(p_0)$ for some $S\in \mathcal{S}$.
	By the first observation, this implies that
	$\pi^{\ell}(p_0)= \{S,p_0\}$.
	This means that $p_1,\ldots,p_6$ have formed coalitions among each other.
	Yet, as argued by \citet[Example~1]{ABS11c},
	there is no CS partition for these six agents.
	Note that resent, which influences each player by a total utility change of at most $\frac 92t+3$,
	does not facilitate stability because the differences between the utilities in 
	this group are at least $10t$.
	It follows that $\pi^{\ell}$ is not CS which is a contradiction.
	So no set agent is together with~$p_0$ in $\pi^{\ell}$.
	
	Next, since $\pi^{\ell}$ is CS, no coalition $\{S,p_0\}$ with
	$S\in\mathcal{S}$ has a CS deviation.
	Since $p_0$ is together with no set agent in $\pi^{\ell}$,
	$p_0$ can have a utility of at most $10t$ for $\pi^{\ell}$ and
	wants to deviate to $\{S,p_0\}$.
	So, $S$ does not want to deviate to $\{S,p_0\}$.
	Therefore, 
	$u_{S}^{\ell}(\pi^{\ell})\geq
	u_{S}^{\ell}(\{S,p_0\})\geq 60t-(\frac 92t+3)=\frac{111}2 t-3$.
	Hence,
	$\pi^{\ell}(S)$ contains at least one filling agent or three element agents $x$ with $x\in S$.
	Since no two set agents are in the same coalition, 
	this implies the existence of an exact cover for $U$.
	
	For appreciative agents, 
	we can use the same construction and the proof is very similar.
\end{proof}

\end{document}